\documentclass[aps,pra,amsmath,amssymb,twocolumn,showpacs]{revtex4} %tirei o twocolumn e showpacs
\pdfoutput=1

\usepackage{graphicx}% Include figure files
\usepackage{dcolumn}% Align table columns on decimal point
\usepackage{bm}% bold math
\usepackage{bbold}% bold math
\usepackage{pslatex}
\usepackage{epsfig}
\usepackage{color}
\usepackage{braket}
\usepackage{amsmath}
\usepackage{amsthm}
\usepackage{amsfonts}
\usepackage{amssymb}
\usepackage{float}
\usepackage[pdfpagelabels,plainpages=false,pdfborder={0 0 0},colorlinks,urlcolor=blue,linkcolor=blue,citecolor=blue,filecolor=black]{hyperref}

\newcommand{\be}{\begin{equation}}
\newcommand{\ee}{\end{equation}}
\newcommand{\ba}{\begin{eqnarray}}
\newcommand{\ea}{\end{eqnarray}}

\newtheorem{thm}{Theorem}

\newtheorem{lem}{Lemma}
\newtheorem{defini}{Definition}
\renewenvironment{proof}{\textit{Proof.}}{{$\hfill\blacksquare$}}

\makeatletter
\renewcommand*\env@matrix[1][*\c@MaxMatrixCols c]{%
  \hskip -\arraycolsep
  \let\@ifnextchar\new@ifnextchar
  \array{#1}}
\makeatother

\begin{document}

\preprint{PRE/003}

\title{Classifying, quantifying, and witnessing qudit-qumode hybrid entanglement}

\author{Karsten Kreis}%
 \email{karsten.kreis@mpl.mpg.de}

\author{Peter van Loock}%
\email{peter.vanloock@mpl.mpg.de}

\affiliation{Optical Quantum Information Theory Group, Max Planck Institute for the Science of Light, G\"unther-Scharowsky-Str. 1/Bau 26, D-91058 Erlangen, Germany}
\affiliation{Institute of Theoretical Physics I, Universit\"at Erlangen-N\"urnberg, Staudtstr. 7/B2, D-91058 Erlangen, Germany}

\date{\today}% It is always \today, today,


\begin{abstract}
Recently, several hybrid approaches to quantum information emerged which utilize both continuous- and discrete-variable methods and resources at the same time. In this work, we investigate the bipartite hybrid entanglement between a finite-dimensional, discrete-variable quantum system and an infinite-dimensional, continuous-variable quantum system. A classification scheme is presented leading to a distinction between pure hybrid entangled states, mixed hybrid entangled states (those effectively supported by an overall finite-dimensional Hilbert space), and so-called truly hybrid entangled states (those which cannot be described in an overall finite-dimensional Hilbert space). Examples for states of each regime are given and entanglement witnessing as well as quantification are discussed. In particular, using the channel map of a thermal photon noise channel, we find that true hybrid entanglement naturally occurs in physically important settings. Finally, extensions from bipartite to multipartite hybrid entanglement are considered.
\end{abstract}

\maketitle

\section{Introduction}
In quantum-information science, the highly intriguing and nonclassical phenomenon of entanglement is at the heart of virtually all applications \cite{4Horodeckis}. Being first considered by Einstein, Podolski, and Rosen, presenting the so-called EPR-paradox \cite{EPR}, entanglement has been in the focus of much research, especially in the recent past. However, there are still a lot of open questions remaining. When dealing with entanglement, one normally deals either with purely discrete-variable (DV) quantum states which live in a finite-dimensional Hilbert space, or with fully continuous-variable (CV), infinite-dimensional quantum systems. The toolbox for analysis is quite different for the two settings. While in the DV case, density matrices provide a complete and convenient representation, the investigation of CV states, due to their infinite dimensionality, is more subtle. However, at least in the Gaussian case, finitely many first and second moments are sufficient for a compact and complete representation \cite{Plenio,Eisert,Adesso,Paris,Braunstein}. 

Recently, so-called \textit{hybrid} protocols have emerged which utilize both CV and DV resources at once \cite{PvLx,Furu,Spagnolo,XWang,SLloyd}. These resources may include CV and DV states as well as CV and DV quantum operations and measurement techniques. It is worth noting that the term \textit{hybrid}, which seems to be quite in vogue at the moment, is used in different contexts in quantum information. For example, there are proposals considering \textit{hybrid} quantum devices which combine elements from atomic and molecular physics as well as from quantum optics and also solid-state physics \cite{Wallquist}. Furthermore, there is the notion of \textit{hybrid} entanglement, referring to entanglement between different degrees of freedom, for example, the entanglement between spatial and polarization modes \cite{Neves,Gabriel}. However, in the present work, we use \textit{hybrid} in the above first-mentioned sense and define \textit{hybrid entanglement} as the entanglement between a finite-dimensional, DV quantum system and an infinite-dimensional, CV quantum system. The prime example is an entangled state between an atomic spin and an electromagnetic mode. As already pointed out, the description of CV and DV states would typically differ. Hence, combining CV and DV quantum systems has its own characterizing and challenging subtleties. So, why would it be useful to consider such hybrid approaches?

In CV quantum computation, Gaussian states as well as Gaussian transformations, such as beam splitting and squeezing, are used. However, to reach computational universality just linear Gaussian elements are not sufficient \cite{SLloyd2}. At least one non-Gaussian component is necessary. Actually, any quantum computer utilizing only linear elements could be efficiently simulated by a classical computer \cite{Bartlett}. This single non-Gaussian element is the main challenge in CV quantum computation, as it is very difficult to efficiently realize such non-Gaussian transformations. In DV quantum computation with photons, the encoding of information takes place in a finite-dimensional subspace of the infinite-dimensional Fock space. Just as in the CV case, for deterministic processing, a nonlinear interaction is required to realize DV universality \cite{Lutkenhaus}. But when truncating the Fock space, only single or few-photon states are left. The drawback of optical DV quantum computation then is that nonlinear interactions on the few-photon level are hard to achieve. Note that a well-known efficient protocol for universal DV computation with only linear optics is still probabilistic (or near-deterministic at the expense of complicated states entangled between sufficiently many photons) \cite{KLM}.

A way out of the problems of CV and DV quantum computation may be provided by those hybrid approaches. The scheme by Gottesman, Kitaev, and Preskill, which makes use of CV Gaussian states and transformations in combination with DV photon number measurements, can be considered one of the first hybrid protocols for quantum computation \cite{GKP}. So-called non-Gaussian phase states are created from Gaussian two-mode squeezed states with the aid of photon counting measurements. Additionally, the protocol can be considered as hybrid, as it employs the concept of encoding logical DV qubits into CV harmonic oscillator modes (qumodes).

However, there are other situations in which one may benefit from combining CV and DV techniques. For example, hybrid entangled states can be exploited for the generation of coherent-state superpositions within the framework of cavity QED \cite{Savage}. Furthermore, there are quantum key distribution schemes which make use of hybrid entanglement \cite{Lorenz,Rigas,Wittmann}. Finally, so-called qubus, i.e., quantum-bus-based, schemes have been developed for establishing entanglement between distant qubits. These also involve hybrid entanglement \cite{qubus1,qubus2,qubus3}.

We can conclude that hybrid entanglement is a key ingredient of various recent quantum-information protocols. In this paper, we perform a thorough classification of hybrid entangled states with the focus on bipartite hybrid entanglement. Nevertheless, also multipartite hybrid entangled states shall be briefly discussed at the end of the paper. In addition to the derivation of a complete classification scheme, which distinguishes between DV-like hybrid entanglement and true hybrid entanglement, also entanglement witnessing and quantification of hybrid entanglement will be discussed.

The paper is organized as follows. In Sec. \ref{sec:2}, we derive the classification scheme for bipartite hybrid entangled states. First, the non-Gaussianity of hybrid entanglement is proved.
Then we introduce an inverse Gram-Schmidt process, which is used for the derivation of the classification scheme. In Sec. \ref{sec:3}, examples for each class of hybrid entangled states are presented. Section \ref{sec:4} briefly discusses multipartite hybrid entanglement, before we give a conclusion in Sec. \ref{sec:5}. Appendices \ref{Apx:PrfLemma} and \ref{Apx:AuxCal} present auxiliary calculations.

\section{Classifying bipartite hybrid entanglement}\label{sec:2}
Let us start with a rigorous definition of bipartite hybrid entanglement.
\begin{defini}
Any entangled bipartite state of the form
\begin{equation}\label{eq:HEdef}
\begin{aligned}
\hat{\rho}^{AB} &=\sum_{n=1}^N p_n\,\ket{\psi_n}_{AB}\bra{\psi_n}\,,\quad p_n>0\;\forall \;n\,;\;\sum_{n=1}^N \;p_n=1, \\
\ket{\psi_n} &=\sum_{m=0}^{d-1} c_{nm}\ket{m}^A\ket{\psi_{nm}}^B\,,\\
c_{nm}&\in\mathbb C\;\forall \;n,m\,;\;\sum_{m=0}^{d-1} \;|c_{nm}|^2=1,
\end{aligned} 
\end{equation}
for $1\leq N\leq\infty$ with generally nonorthogonal qumode state vectors $\ket{\psi_{nm}}^B$, defined in the total Hilbert space ${\mathcal H}^{AB}={\mathcal H}_{d}^A\otimes{\mathcal H}_{\infty}^B$ with finite $d\geq2$, is called \textup{hybrid entangled}.
\end{defini}
Since every bipartite hybrid entangled state can be written in such a pure-state decomposition, this definition is complete. Note that later, in Sec. \ref{subsec:2.c}, the rank parameter $N$ will determine one of three possible classes of hybrid entangled states.

\subsection{Non-Gaussianity}\label{subsec:2.a}
How can hybrid entangled states actually be described in a convenient way? Can their entanglement be
quantified? For an overall infinite-dimensional Hilbert space ${\mathcal H}_{d}^A\otimes{\mathcal H}_{\infty}^B$, using the standard finite-dimensional techniques, density matrices can no longer be employed. So, the use of CV methods can be attempted. However, in CV entanglement theory, the only conveniently representable states are the Gaussian ones. In the non-Gaussian CV regime, for instance, exact entanglement quantification is, in general, hard to achieve. It is therefore useful to ask whether hybrid entangled Gaussian states exist.
\begin{lem}\label{lem:DVnonG}
Any single-partite $d$-dimensional quantum state with finite $d$ and $d\geq2$ is non-Gaussian.
\end{lem}
For the proof of Lemma \ref{lem:DVnonG}, see Appendix \ref{Apx:PrfLemma}.
\begin{thm}\label{thm:hybridnonG}
Any bipartite hybrid entangled or classically correlated state is non-Gaussian \footnote[1]{Questions concerning the quantum discord \cite{Ollivier} of (mixed) bipartite hybrid states, and the Gaussian or non-Gaussian nature of such states, are left for future research.}. 
\end{thm}
\begin{proof}
If a multipartite quantum state is Gaussian, all its subsystems will be Gaussian. So, we consider a state of the form \eqref{eq:HEdef} and trace out the CV subsystem. What is left is a $d$-dimensional single-partite system with finite $d$, which can be described in its Fock basis. It is denoted by $\hat{\rho}$. Due to Lemma \ref{lem:DVnonG}, $\hat{\rho}$ is Gaussian if and only if $\hat{\rho}=\ket{0}\bra{0}$, which is a pure state. However, for any entangled or even classically correlated state, the reduced state cannot be pure \cite{Audretsch}. Hence, for Gaussian $\hat{\rho}$, the overall system cannot be entangled or classically correlated. Therefore, every bipartite hybrid entangled or classically correlated state is non-Gaussian.
\end{proof}

We have found that there are no hybrid entangled Gaussian states. Hence, CV Gaussian tools are inappropriate for the description of hybrid entangled states.

The proof of Theorem \ref{thm:hybridnonG} basically relies on Lemma \ref{lem:DVnonG}, which states that any finite-dimensional DV state with dimension $\geq2$ is non-Gaussian. It is straightforwardly generalized to multipartite systems. Arguing that for Gaussianity \textit{all} subsystems have to be Gaussian, it is sufficient for the non-Gaussianity of any multipartite system which also possesses DV constituent(s) that at least one DV subsystem is of dimension $\geq2$. However, if there shall be entanglement between the DV subsystem and the rest, the state will necessarily be non-Gaussian, since for entanglement dimension, $\geq2$ is required. So, any multipartite quantum state which involves entangled DV subsystems is non-Gaussian. Even more generally, only 1- or infinite-dimensional systems can be Gaussian. This is, of course, not surprising; however, it is worth pointing out that Theorem \ref{thm:hybridnonG} is not trivial. Since the \textit{overall} Hilbert space of hybrid entangled systems is indeed infinite-dimensional, it is not \textit{a priori} clear that Gaussian hybrid entangled states do not exist.

We have shown that hybrid entangled states belong to the non-Gaussian, infinite-dimensional Hilbert-space regime, which is not easy to deal with, as we already know from conventional CV entanglement theory. The states can be neither described by proper density matrices nor by covariance matrices. Phase-space representations are also not so convenient, since one of the subsystems is DV. The only known quasi-probability-distribution which may in some cases make direct statements about the separability properties of the state is the Glauber-Sudarshan \textit{P} representation \cite{WangX,Duan}. However, it can be easily shown that this function is totally irregular for these highly nonclassical hybrid entangled states. It may still be possible to construct entanglement witnesses, but entanglement quantification appears, in general, hard. Actually, there is a way out of this dilemma: For some hybrid entangled states, the unique Hilbert-space structure can be exploited in such a way that the states can nevertheless be described by density matrices. These states are effectively finite-dimensional and are therefore called DV-like hybrid entangled. This gives rise to a classification scheme of hybrid entangled states. However, for a rigorous analysis we first need to introduce a slightly modified version of the so-called Gram-Schmidt process.

\subsection{Inverse Gram-Schmidt process}\label{subsec:2.b}
Consider again a general hybrid entangled state of the form \eqref{eq:HEdef}. Depending on the number of mix terms $N$ and the dimension of the DV subsystem $d$, there can be maximally $N\times d$ linearly independent CV qumode states $\ket{\psi_{nm}}^B$ in $\hat{\rho}^{AB}$. The dimension $d$ is always finite due to the definition of hybrid entanglement. However, $N$ may be either finite or infinite and, hence, the number of linearly independent qumode states is either finite or infinite. Furthermore, if $N=1$, the state is pure.

If the number $N\times d$ of linearly independent CV qumode states is finite, they only span an $(N\times d)$-dimensional subspace ${\mathcal H}_{N\times d}$ of the initially infinite-dimensional Hilbert space ${\mathcal H}_\infty$. Then, the \textit{Gram-Schmidt process} can be employed to express the qumode states in an orthonormal basis of this finite-dimensional subspace. In this case, the state becomes effectively DV and all the methods from DV entanglement theory can be applied. 

The Gram-Schmidt procedure is a method for orthonormalizing a finite, linearly independent set of vectors in an inner product space \cite{Nielsen}. For a linearly independent set of vectors $\{\ket{\psi_i}:i=1,\ldots,n\}$ (since the process is to be exploited in the framework of Hilbert spaces, the inner product space is \textit{a priori} assumed to be a Hilbert space), a set of pairwise orthonormal vectors $\{\ket{e_i}:i=1,\ldots,n\}$ spanning the same subspace as $\{\ket{\psi_i}:i=1,\ldots,n\}$ is given by
\begin{equation}\label{GSprocess}
\begin{aligned}
\ket{e_1'} &= \ket{\psi_1}\,, & \ket{e_1} &= \frac{\ket{e_1'}}{\sqrt{\braket{e_1'|e_1'}}}\,, \\
\ket{e_2'} &= \ket{\psi_2}-\braket{e_1|\psi_2}\ket{e_1}\,, & \ket{e_2} &= \frac{\ket{e_2'}}{\sqrt{\braket{e_2'|e_2'}}}\,, \\
& \vdots & & \vdots \\
\ket{e_n'} &=\ket{\psi_n}-\sum_{i=1}^{n-1}\braket{e_i|\psi_n}\ket{e_i}\,,\qquad & \ket{e_n} &=\frac{\ket{e_n'}}{\sqrt{\braket{e_n'|e_n'}}}\,.
\end{aligned}
\end{equation}
Making use of this, any finite set of linearly independent qumode states can be expressed in an orthonormal basis $\{\ket{e_i}:i=1,\ldots,n\}$. However, more specifically, Eqs. \eqref{GSprocess} only determine how to express the new orthonormal basis in terms of the old nonorthonormal one. What is actually required is the inverse expression. To express the qumode states in terms of $\{\ket{e_i}:i=1,\ldots,n\}$ we try the following approach, which can be considered an \textit{inverse Gram-Schmidt process}, i.e., a modified version of the original Gram-Schmidt process.
\begin{thm}\label{thm:GSI}
$n$ normalized, generally nonorthogonal, linearly independent states $\{\ket{\psi_i}:i=1,\ldots,n \,;\,0\leq|\braket{\psi_i|\psi_j}|\leq1\,\forall\,i,j\}$ can always be expressed as $\ket{\psi_i}=\sum_{j=1}^ia_{ij}\ket{e_j}$, where $\{\ket{e_i}:i=1,\ldots,n\}$ forms an orthonormal basis of the space spanned by $\{\ket{\psi_i}\}$, and $a_{ij}\in\mathbb C$.
\end{thm}
\begin{proof}
Writing out $\ket{\psi_i}=\sum_{j=1}^ia_{ij}\ket{e_j}$, Theorem \ref{thm:GSI} states that the $\{\ket{\psi_i}\}$ can be always written as
\begin{equation}\label{GSprocessInverse}
\begin{aligned}
\ket{\psi_1} &= a_{11}\ket{e_1}\,, \\
\ket{\psi_2} &= a_{21}\ket{e_1}+a_{22}\ket{e_2}\,, \\
\ket{\psi_3} &= a_{31}\ket{e_1}+a_{32}\ket{e_2}+a_{33}\ket{e_3}\,, \\
& \vdots  \\
\ket{\psi_n} &= a_{n1}\ket{e_1}+\ldots+a_{nn}\ket{e_n}=\sum_{i=1}^{n}a_{ni}\ket{e_i}\,.
\end{aligned}
\end{equation}
For the proof, it has to be shown that (1) $\ket{\psi_i}=\sum_{j=1}^ia_{ij}\ket{e_j}$ corresponds to a valid basis transformation and (2) it actually performs the right mapping.

(1) Write the transformation as
\begin{equation}
\ket{\psi_i}=\sum_j T_{ij}\ket{e_j},
\end{equation}
with the transformation matrix 
\begin{equation}
T=\begin{pmatrix} a_{11} & 0 & \cdots & 0 \\ 
a_{21} & a_{22} & & \vdots \\ 
\vdots & & \ddots & \vdots \\ 
a_{n1} & a_{n2} & \cdots & a_{nn}
\end{pmatrix}.
\end{equation}
Due to the normalization of the initial and the new vectors $\sum_{j=1}^i |a_{ij}|^2=1$ is known, and due to the linear independence of the $\{\ket{\psi_i}\}$ also $a_{ii}\neq0\,\forall\,i$ is necessary.
\begin{itemize}
\item[$\Rightarrow$] $\det[T] =\prod_{i=1}^n a_{ii}\neq0$. 
\item[$\Rightarrow$] $T$ is invertible.
\item[$\Rightarrow$] $T$ is a valid basis transformation.
\end{itemize}

(2) To show that the lower triangular structure of the basis transformation $T$ in combination with the orthonormal basis $\{\ket{e_i}\}$ is sufficient to actually express the $\{\ket{\psi_i}\}$ accurately in terms of $\{\ket{e_i}\}$, it is demonstrated that the $\frac{n(n+1)}{2}$ parameters $a_{ij}$ can be chosen such that all overlaps $\braket{\psi_i|\psi_j}$ are preserved when the transformation is applied.

On the one hand, there are $n^2$ such overlaps in total and $\frac{n(n+1)}{2}$ ones with potentially differing absolute values. On the other hand, there are $\frac{n(n+1)}{2}$ complex parameters $a_{ij}$. From the structure of the basis transformation and the fact that $a_{ij}$ are complex, it is clear that if $a_{ij}$ can be chosen such that the $\frac{n(n+1)}{2}$-element set of $\{\braket{\psi_i|\psi_j}:i\leq j\}$ can be preserved, also the rest of the overlaps are preserved, since they are only complex conjugates of the former. Hence, it already becomes reasonable that the $a_{ij}$ can be chosen appropriately.

However, a proper proof is performed by induction in $n$:

\textit{Inductive Basis}: $n=1$. There is only one overlap to be preserved: 
\begin{equation}
\braket{\psi_1|\psi_1}=1\stackrel{!}{=}\braket{e_1|a_{11}^\ast a_{11}|e_1}=|a_{11}|^2.
\end{equation}
Hence, choose $a_{11}=1$, which preserves the overlap $\braket{\psi_1|\psi_1}$.

\textit{Inductive Step}: Assume $a_{ij}$ have been calculated for $i\leq n-1$ such that all overlaps $\{\braket{\psi_i|\psi_j}:i,j\leq n-1\}$ are preserved. We show that then also all $a_{nj}$ can be chosen such that the overlaps $\{\braket{\psi_i|\psi_n}:i=1,\ldots,n\}$ are preserved. The complex conjugated overlaps follow automatically as argued before.

Applying the basis transformation, the overlaps $\{\braket{\psi_i|\psi_n}:i=1,\ldots,n-1\}$ are
\begin{equation}
\begin{aligned}
\braket{\psi_1|\psi_n} &= a_{11}^\ast a_{n1}\,, \\
\braket{\psi_2|\psi_n} &= a_{21}^\ast a_{n1}+a_{22}^\ast a_{n2}\,, \\
& \vdots  \\
\braket{\psi_{n-1}|\psi_n} &= a_{n-1,1}^\ast a_{n1} + \ldots + a_{n-1,n-1}^\ast a_{n,n-1}  \,.
\end{aligned}
\end{equation}
As $a_{ij}$ have been calculated for $i\leq n-1$ due to the inductive hypothesis, this is just a system of linear equations, which can be written as an augmented matrix (using $a_{11}^\ast=1$):
\begin{equation}\label{GSIaug}
\begin{pmatrix}[c|cccc]
\braket{\psi_1|\psi_n}  & 1 & 0 & \hdots & 0 \\
\braket{\psi_2|\psi_n}  & a_{21}^\ast & a_{22}^\ast & & \vdots \\
 & \vdots & & \ddots & \vdots  \\
\braket{\psi_{n-1}|\psi_n}  & a_{n-1,1}^\ast & a_{n-1,2}^\ast & \hdots & a_{n-1,n-1}^\ast \end{pmatrix}.
\end{equation}
Since $a_{ii}\neq0\,\forall\,i$, this system of equations is exactly solvable. Hence, $\{a_{nj}:j=1,\ldots,n-1\}$ can be chosen such that $\{\braket{\psi_i|\psi_n}:i=1,\ldots,n-1\}$ are preserved. Therefore, there are only one free parameter $a_{nn}$ and one overlap $\braket{\psi_n|\psi_n}$ to be preserved left:
\begin{equation}
\braket{\psi_n|\psi_n}=1\stackrel{!}{=}\sum_{j=1}^n |a_{nj}|^2=\sum_{j=1}^{n-1} |a_{nj}|^2 + |a_{nn}|^2.
\end{equation}
From $\sum_{j=1}^n |a_{nj}|^2=1$ and $a_{nn}\neq 0$, which is already known, $\sum_{j=1}^{n-1} |a_{nj}|^2<1$ follows, and hence $a_{nn}$ can be chosen as  
\begin{equation}
a_{nn}=\sqrt{1-\sum_{j=1}^{n-1} |a_{nj}|^2}.
\end{equation}
In the end, also $\braket{\psi_n|\psi_n}$ can be preserved. Therefore, the theorem is valid for $n$ under the assumption of validity for $n-1$.  
As a conclusion, with the inductive basis it is valid for all $n$.
\end{proof}

The proof has been presented in such a great detail, because it is constructive and hence also sets out how to actually compute $\{a_{ij}\}$ for a given set of qumode states. Recalling Eqs. \eqref{GSprocessInverse} and \eqref{GSIaug}, $a_{ij}$ can be calculated successively one after another by considering successive overlaps. A parameter $a_{i1}$ is directly obtained from the overlap $\braket{\psi_1|\psi_i}$. Then, $a_{i2}$ follows from $\braket{\psi_2|\psi_i}$ together with the known $a_{i1}$. Likewise, $a_{i3}$ is calculated from $\braket{\psi_3|\psi_i}$, $a_{i1}$ and $a_{i2}$. For the other parameters just go on like this. Hence, the inverse Gram-Schmidt process can be efficiently implemented and computed. 

As an example consider the normalized qumode states $\{\ket{\psi_i}:i=1,2,3\}$ with overlaps
\begin{equation}
\begin{aligned}
\braket{\psi_1|\psi_2} &= c_1, \\
\braket{\psi_1|\psi_3} &= c_2, \\
\braket{\psi_2|\psi_3} &= c_3.
\end{aligned}
\end{equation} 
They can be expressed in an orthonormal basis $\{\ket{e_i}:i=1,2,3\,,\,\braket{e_i|e_j}=\delta_{ij}\}$ as
\begin{equation}\label{eqGSIinv}
\begin{aligned}
\ket{\psi_1} =& \ket{e_1}, \\
\ket{\psi_2} =& c_1\ket{e_1}+\sqrt{1-|c_1|^2}\ket{e_2}, \\
\ket{\psi_3} =& c_2\ket{e_1}+\frac{c_3-c_1^\ast c_2}{\sqrt{1-|c_1|^2}}\ket{e_2} \\
&+\sqrt{1-|c_2|^2-\frac{|c_3-c_1^\ast c_2|^2}{1-|c_1|^2}}\ket{e_3}.
\end{aligned}
\end{equation}

\subsection{Classification scheme}\label{subsec:2.c}
So, once again, consider a state of the form \eqref{eq:HEdef}. For $N=1$ the state is pure. Then it contains only $d$ linearly independent qumode states, which span a $d$-dimensional subspace ${\mathcal H}_d^B$. Therefore, with aid of the inverse Gram-Schmidt process these qumode states can be expressed in an orthonormal basis. Then, density matrices can be employed for the description of the overall state, and also a Schmidt decomposition can be performed or pure-state measures such as the entropy of entanglement can be calculated \cite{Audretsch,Bennett}. The state is then effectively DV. For $1<N<\infty$, the state is mixed. Nevertheless, it possesses a finite number of $N\times d$ qumode states, which can be again cast in an orthonormal basis. Therefore, also states of this kind are effectively DV and the density matrix formalism can be exploited. However, such states are no longer pure and neither pure-state measures nor a Schmidt decomposition can be applied. Finally, there is the case $N=\infty$. Here, $N=\infty$ refers to those states which can be expressed with infinite $N$ only. %and not with finite $N$ also. Remember that pure state decompositions are not unique. Anyway, 
These states hold an infinite number of qumode states, which has the effect that the Gram-Schmidt process cannot be applied anymore. Hence, they are not describable by density matrices and therefore no longer effectively DV. Nevertheless, one subsystem does remain DV. These are the states which we call \textit{truly hybrid entangled}.

Summing up, for bipartite hybrid entangled states in a pure-state decomposition with $N$ denoting the number of mix terms in the convex combination of pure-state projectors, there is the following classification scheme:
\begin{itemize}
\item $\boldsymbol{N=1:}$ Pure hybrid entangled states.
\begin{itemize}
\item Supported by a finite-dimensional subspace. 
\item $\,\Rightarrow\,$ \textbf{DV-like entanglement}.
\item Schmidt decomposition applicable.
\item DV pure-state measures applicable.
\end{itemize}
\item $\boldsymbol{1<N<\infty:}$ Mixed hybrid entangled states.
\begin{itemize}
\item Supported by a finite-dimensional subspace. 
\item $\,\Rightarrow\,$ \textbf{DV-like entanglement}.
\item DV mixed-state measures applicable.
\end{itemize}
\item $\boldsymbol{N=\infty:}$ Mixed hybrid entangled states.
\begin{itemize}
\item No support by a finite-dimensional subspace.
\item $\,\Rightarrow\,$ \textbf{True hybrid entanglement}.
\item In general no exact measures directly applicable.
\item CV entanglement witnesses adaptable.
\end{itemize}
\end{itemize}
It should be pointed out that the possibility of applying DV methods on such a wide class of hybrid entangled states is quite remarkable. An initially non-Gaussian infinite-dimensional quantum state, which seems rather awkward at first sight, can finally be conveniently described in terms of density matrices, and in the pure-state case even a Schmidt decomposition can be performed. Nevertheless, there is also the class of states which stay \textit{truly hybrid entangled} and cannot be transformed using the Gram-Schmidt process. As mentioned earlier, in this infinite-dimensional, non-Gaussian regime exact entanglement quantification appears to be hard. However, entanglement witnesses for detecting true hybrid entanglement can be adapted from CV entanglement theory.

\section{Examples and Applications}\label{sec:3}
Now, we shall present an example for each class of hybrid entanglement. We shall also point out when these examples correspond to manifestations of entangled states encountered in quantum-information protocols and applications.

\subsection{Pure qutrit-qumode entanglement}\label{subsec:3a}
As an example of pure bipartite hybrid entanglement an entangled state of a qutrit and a qumode system is considered:
\begin{equation}\label{eq:pQtQmS}
\ket{\psi}^{AB} =\frac{1}{\sqrt{3}}\Bigl(\ket{e_0}^A\ket{vac}^B+\ket{e_1}^A\ket{\alpha}^B+\ket{e_2}^A\ket{-\alpha}^B\Bigr).
\end{equation}

(1) An inverse Gram-Schmidt process with respect to subsystem B yields
\begin{align}
\ket{vac}^B  =& \;\;\,\ket{e_0}^B,  \\
\ket{\alpha}^B  =& x\ket{e_0}^B + \quad\sqrt{1-x^2}\ket{e_1}^B, \\
\ket{-\alpha}^B  =& x\ket{e_0}^B - x^2\sqrt{1-x^2}\ket{e_1}^B \notag \\
&+ \sqrt{1-x^2-x^4+x^6}\ket{e_2}^B,
\end{align}
with $x=e^{-\frac{1}{2}|\alpha|^2}$. This corresponds to Eq. \eqref{eqGSIinv}. Hence
\begin{equation}
\begin{aligned}
\ket{\psi}^{AB} =& \frac{1}{\sqrt{3}}\Bigl(\ket{e_0}^A\ket{e_0}^B+x\ket{e_1}^A\ket{e_0}^B+\sqrt{1-x^2}\ket{e_1}^A\ket{e_1}^B \\ 
&+x\ket{e_2}^A\ket{e_0}^B-x^2\sqrt{1-x^2}\ket{e_2}^A\ket{e_1}^B \\ 
&+\sqrt{1-x^2-x^4+x^6}\ket{e_2}^A\ket{e_2}^B\Bigr),
\end{aligned}
\end{equation}
which is the effective DV form of the state.

(2) The pure-state entropy of entanglement of the state can be calculated, which is shown in Fig. \ref{fig:pQtQmS}. For $\alpha\rightarrow\infty$ the state becomes maximally entangled ($\alpha$ real). 
\begin{figure}[ht]
\begin{center}
\includegraphics[width=8.5cm]{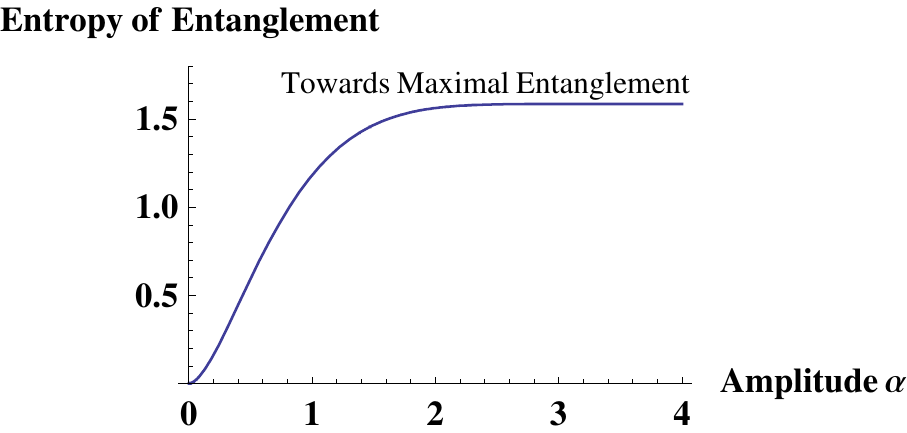}
\caption[Entropy of entanglement of the state $\ket{\psi}^{AB}$, \eqref{eq:pQtQmS}.]{(Color online) Entropy of entanglement of the state $\ket{\psi}^{AB}$ of Eq. \eqref{eq:pQtQmS}. For $\alpha\rightarrow\infty$ the state becomes maximally entangled. Note that the entropy of entanglement has been calculated here in qubit entanglement units (ebits: $log\equiv log_2$).}
\label{fig:pQtQmS}
\end{center}
\end{figure}

(3) Finally, set $x=\frac{1}{2}$, which corresponds to $\alpha=\sqrt{2\ln2}\approx1.18$, and calculate the Schmidt decomposition:
\begin{equation}
\begin{aligned}
\ket{\psi}^{AB}= &0.76\ket{e_0'}^A\ket{e_0'}^B+0.56\ket{e_1'}^A\ket{e_1'}^B \\
& +0.33\ket{e_2'}^A\ket{e_2'}^B.
\end{aligned}
\end{equation}
We can also think of other examples for pure bipartite hybrid entangled states. States of the form 
\begin{equation}
\ket{\psi}^{AB} =\frac{1}{\sqrt{2}}\Bigl(\ket{e_0}^A\ket{\psi_0}^B+\ket{e_1}^A\ket{\psi_1}^B\Bigr)
\end{equation}
are pure bipartite hybrid entangled and relevant for cat-state engineering \cite{Savage,PvLx}, hybrid quantum communication via qubus approaches \cite{qubus1,qubus3}, and some quantum key distribution schemes \cite{Lorenz,Rigas} [later, in Sec. \ref{subsec:3c}, we obtain such a state in Eq. \eqref{eq:noisyQbQm1norm} as a special case of Eq. \eqref{eq:noisyQbQm3} with $\braket{n_{th}}=0$ and $\eta=1$].

\subsection{Mixed qubit-qumode entanglement}\label{subsec:3b}
Now, we present an example for the class of mixed, but effectively DV hybrid entangled states. Consider the state
\begin{equation}\label{eq:2x3Mstate1}
\begin{aligned}
\hat{\rho}^{AB} &= p\,\ket{\phi_+}^{AB}\bra{\phi_+}+(1-p)\,\ket{\phi_-}^{AB}\bra{\phi_-}, \\ 
\ket{\phi_\pm}^{AB} & = \frac{1}{\sqrt{2}}\Bigl(\ket{e_0}^A\ket{vac}^B+\ket{e_1}^A\ket{\pm\alpha}^B\Bigr),
\end{aligned}
\end{equation}
which contains three qumode states $\ket{vac}^B$ and $\ket{\pm\alpha}^B$ with $\alpha\in{\mathbb R}$ (see Fig. \ref{fig:3x2state} for a visualization).
\begin{figure}[ht]
\begin{center}
\includegraphics[width=8.5cm]{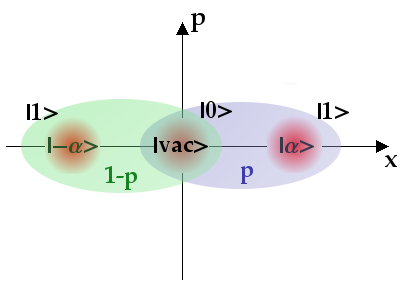}
\caption[Visualization of the hybrid entangled state $\hat{\rho}^{AB}$, see \eqref{eq:2x3Mstate1}, in phase space.]{(Color online) Visualization of the hybrid entangled state $\hat{\rho}^{AB}$ of Eq. \eqref{eq:2x3Mstate1} in phase space. The blue (right) region corresponds to the first pure state in the convex combination with probability $p$, while the green (left) region represents the pure state obtained with probability $1-p$. The additional $\ket{0}$ and $\ket{1}$ vectors denote the qubit states associated with the qumode states $\ket{vac}$ and $\ket{\pm\alpha}$.}
\label{fig:3x2state}
\end{center}
\end{figure}

We perform a Gram-Schmidt process and obtain a qubit-qutrit entangled state in ${\mathcal H}_2^A\otimes{\mathcal H}_3^B$. The pure states in the convex combination of \eqref{eq:2x3Mstate1}, after the inverse Gram-Schmidt process in an orthonormal basis, look like
\begin{align}
\ket{\phi_+}^{AB} =&  \frac{1}{\sqrt{2}}\Bigl(\ket{e_0}^A\ket{e_0}^B+x\ket{e_1}^A\ket{e_0}^B+\sqrt{1-x^2}\ket{e_1}^A\ket{e_1}^B\Bigr),\\
\ket{\phi_-}^{AB} =&  \frac{1}{\sqrt{2}}\Bigl(\ket{e_0}^A\ket{e_0}^B+x\ket{e_1}^A\ket{e_0}^B-x^2\sqrt{1-x^2}\ket{e_1}^A\ket{e_1}^B \notag \\ 
&+\sqrt{1-x^2-x^4+x^6}\ket{e_1}^A\ket{e_2}^B\Bigr),
\end{align}
with $x=\mathrm{e}^{-\frac{1}{2}|\alpha|^2}$. Entanglement quantification can be performed with the logarithmic negativity $E_N$ (see Fig. \ref{fig:3x2LogNeg}) \cite{Vidal,Eisert2}. For general bipartite mixed states with higher dimension than $2\times 2$, it is basically the only known calculable entanglement monotone, and it is also the most important one for general mixed, effectively DV hybrid entangled states.
\begin{figure}[ht]
\begin{center}
\includegraphics[width=8.5cm]{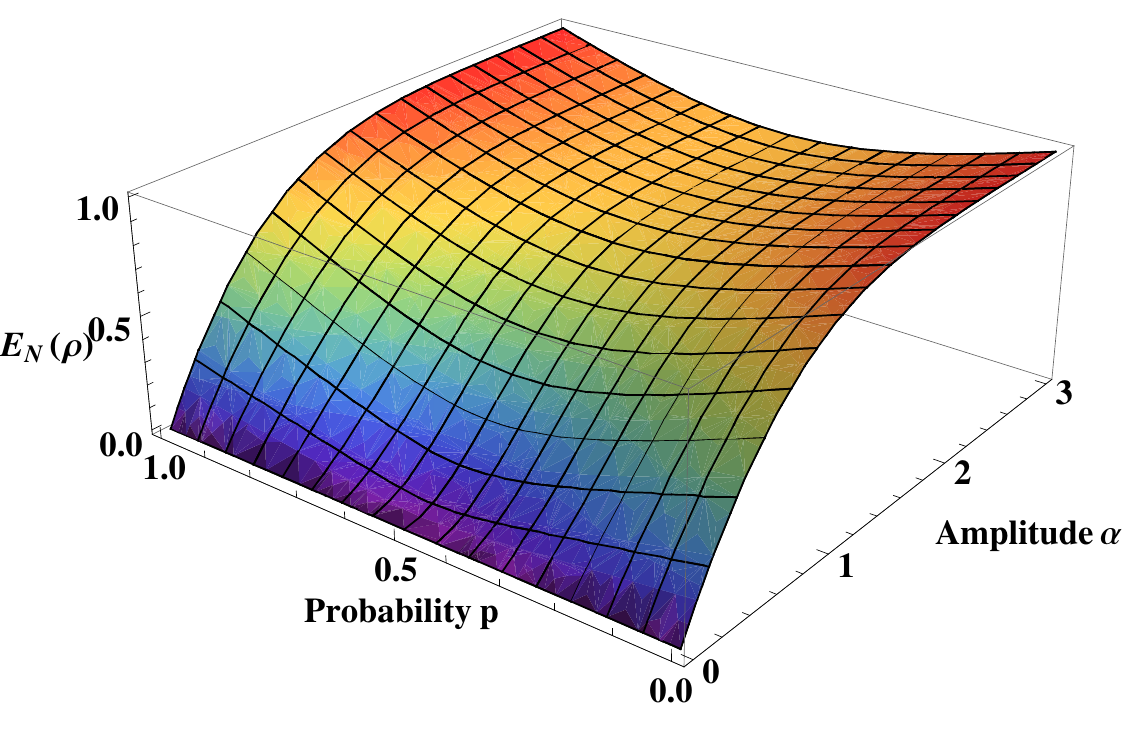}
\caption[Logarithmic negativity of the effective qubit-qutrit hybrid entangled state $\hat{\rho}^{AB}$, Eq. \eqref{eq:2x3Mstate1}, as a function of probability $p$ and amplitude $\alpha$.]{(Color online) Logarithmic negativity of the effective qubit-qutrit hybrid entangled state $\hat{\rho}^{AB}$, Eq. \eqref{eq:2x3Mstate1}, as a function of probability $p$ and amplitude $\alpha$. For any given $\alpha$ maximal mixing $p=\frac{1}{2}$ yields the smallest entanglement. Furthermore, higher $\alpha$ also results in greater entanglement. For $p=0\lor1$ and $\alpha\rightarrow\infty$ the state becomes maximally entangled ($\alpha \in \mathbb R$).}
\label{fig:3x2LogNeg}
\end{center}
\end{figure}
It can be seen that for any given $\alpha$ the smallest entanglement is obtained for maximal mixing $p=\frac{1}{2}$. Furthermore, the higher $\alpha$, the more entangled the state is. Finally, for $p=0\lor1$ and $\alpha\rightarrow\infty$ the state approaches maximal entanglement.

We may also consider other mixed bipartite hybrid entangled states which are even effectively DV-like qubit-qubit entangled. Then the concurrence \cite{Wootters} is a better choice for entanglement quantification. In Sec. \ref{subsec:3c}, such a state is given in Eq. \eqref{Eq0thermCase} as a special case of Eq. \eqref{eq:noisyQbQm3} for $\braket{n_{th}}=0$. Its concurrence is plotted in Fig. \ref{fig:CADampCh}.

\subsection{True hybrid entanglement}\label{subsec:3c}
Finally, the third class of bipartite hybrid entangled states is considered, the \textit{truly hybrid entangled} states. Recall their definition corresponding to Eq. \eqref{eq:HEdef} with $N=\infty$:
\begin{equation}\label{eq:THEdef}
\begin{aligned}
\hat{\rho}^{AB} &=\sum_{n=1}^\infty p_n\,\ket{\psi_n}_{AB}\bra{\psi_n}\,,\qquad p_n>0\;\forall \;n\,,\;\;\;\sum_{n=1}^\infty \;p_n=1, \\
\ket{\psi_n} &=\sum_{m=0}^{d-1} c_{nm}\ket{m}^A\ket{\psi_{nm}}^B\,,\qquad c_{nm}\in\mathbb C\,,\;\;\;\sum_{m=0}^{d-1} \;|c_{nm}|^2=1.
\end{aligned} 
\end{equation}
As can be seen from these equations, truly hybrid entangled states possess an infinite number of qumode states $\ket{\psi_{nm}}^B$. Therefore, the Gram-Schmidt process ceases to work and the states stay in a Hilbert space of the form ${\mathcal H}^A_d\otimes{\mathcal H}^B_\infty$. Hence, they are not effectively DV, but instead really combined $\text{DV}\otimes\text{CV}$ states and therefore \textit{truly} hybrid. Unfortunately, this true ``hybridness'' has the effect that the states live in an overall infinite-dimensional Hilbert space in the non-Gaussian regime. This makes exact entanglement quantification difficult. %Hence, the focus is only on the \textit{detection} of true hybrid entanglement. 
Now, we present an example which is relevant for some quantum information protocols.

Consider the state 
\begin{equation}\label{eq:noisyQbQm1}
\ket{\psi}^{AB}=\frac{1}{\sqrt{2}}\Bigl(\ket{0}^A\ket{\alpha}^B+\ket{1}^A\ket{-\alpha}^B\Bigr),
\end{equation}
which plays a crucial role in cat-state engineering \cite{Savage}, hybrid quantum communication via qubus approaches \cite{qubus1,qubus3}, and some quantum key distribution schemes \cite{Lorenz,Rigas}. It is transmitted through a one-sided thermal photon noise channel, where the noise only affects subsystem B. Writing this out, 
\begin{equation}\label{eq:noisyQbQm2}
\hat{\rho}'{}^{AB}=({\mathbb 1}^A\otimes\$_{thermal}^B)\ket{\psi}^{AB}\bra{\psi}
\end{equation}
is the state to be investigated. The channel is modeled by a beam splitter coupling subsystem B to the environment, which is in a thermal state. Afterward the environment is traced out (see Fig. \ref{fig:NoiseCh}).  
\begin{figure}[ht]
\begin{center}
\includegraphics[width=8.5cm]{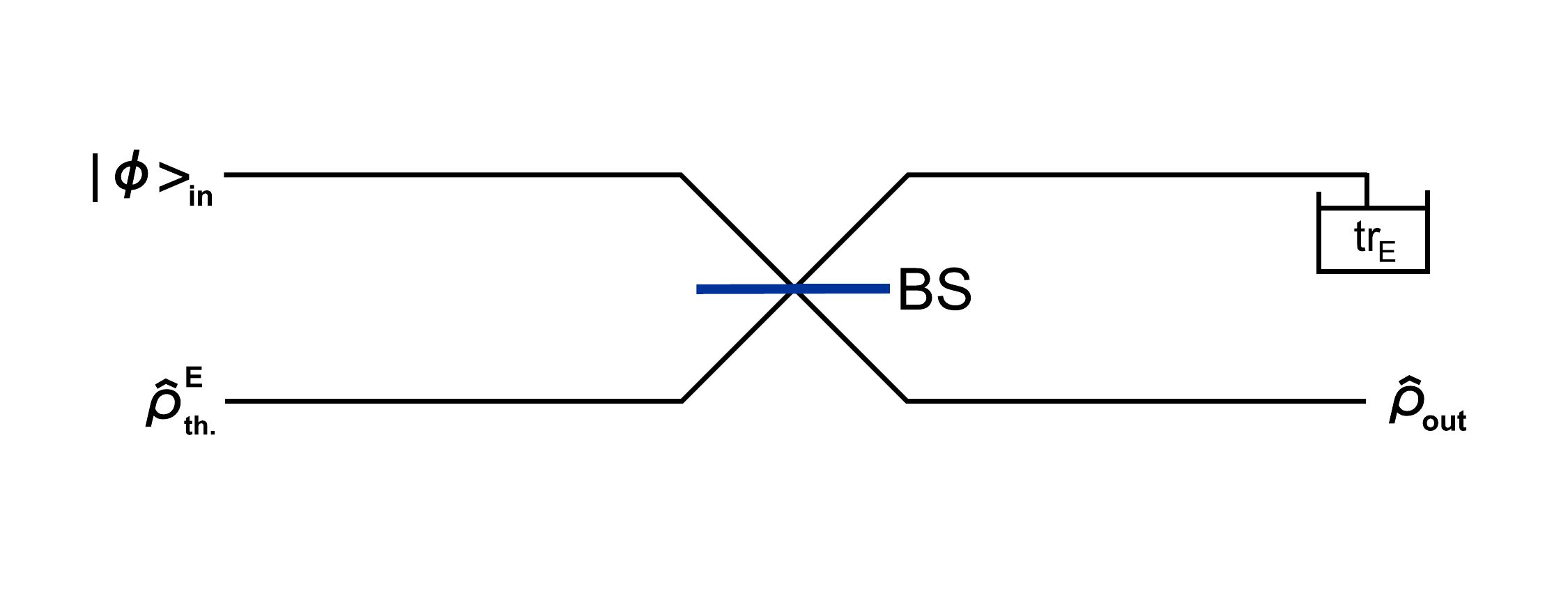}
\caption[Modelling of the photon noise channel.]{(Color online) Modeling of the photon noise channel. The environment mode is in a thermal state and coupled via a beam splitter to the input state to be transmitted. Subsequently the environment mode is traced out and the decohered output state is obtained.}
\label{fig:NoiseCh}
\end{center}
\end{figure}

For the output, 
\begin{equation}\label{eq:noisyQbQm3}
\begin{aligned}
\hat{\rho}'{}^{AB} &=({\mathbb 1}^A\otimes\$_{thermal}^B)\ket{\psi}^{AB}\bra{\psi} \\
&=\frac{1}{2}\sum_{n=0}^\infty \rho_n^{th} \sum_{k,l=0}^n f_{nk}(\eta)f_{nl}(\eta) \\
&\times\biggl(A_{nkl}^{\alpha\alpha}(\eta)\;\ket{0}^A\bra{0}\,\otimes\,\hat{a}^{\dagger^k}\ket{\sqrt{\eta}\alpha}^B\bra{\sqrt{\eta}\alpha}\hat{a}^l  \\
&+ A_{nkl}^{-\alpha-\alpha}(\eta)\;\ket{1}^A\bra{1}\,\otimes\,\hat{a}^{\dagger^k}\ket{-\sqrt{\eta}\alpha}^B\bra{-\sqrt{\eta}\alpha}\hat{a}^l   \\
&+ A_{nkl}^{\alpha-\alpha}(\eta)\;\ket{0}^A\bra{1}\,\otimes\,\hat{a}^{\dagger^k}\ket{\sqrt{\eta}\alpha}^B\bra{-\sqrt{\eta}\alpha}\hat{a}^l \\
&+ A_{nkl}^{-\alpha\alpha}(\eta)\;\ket{1}^A\bra{0}\,\otimes\,\hat{a}^{\dagger^k}\ket{-\sqrt{\eta}\alpha}^B\bra{\sqrt{\eta}\alpha}\hat{a}^l\biggr)
\end{aligned}
\end{equation}
is obtained, where $\hat{a}$ and $\hat{a}^\dagger$ are the mode operators of subsystem B; $\rho_n^{th}$ denotes the thermal photon distribution $\frac{\braket{n_{th}}^n}{(1+\braket{n_{th}})^{n+1}}$ of the environmental thermal state with mean thermal photon number $\braket{n_{th}}$, and $f_{nk}(\eta)$ and $A_{nkl}^{\alpha\beta}(\eta)$ are defined as
\begin{align}
f_{nk}(\eta)&:=\frac{1}{\sqrt{n!}}{\binom n k}\sqrt{\eta}^{n-k}(-\sqrt{1-\eta})^k, \\
A_{nkl}^{\alpha\beta}(\eta)&:=\braket{\sqrt{1-\eta}\beta|\hat{a}^{{}^{n-k}}\hat{a}^{\dagger^{n-l}}|\sqrt{1-\eta}\alpha},
\end{align}
with the beam splitter transmissivity $\eta$ (for a detailed calculation of this result, see Appendix \ref{Apx:AuxCal}). It can be inferred from Eq.  \eqref{eq:noisyQbQm3} that the state $\hat{\rho}'{}^{AB}$ clearly is truly hybrid entangled, as it contains an infinite number of qumode states $\{\hat{a}^{\dagger^k}\ket{\pm\sqrt{\eta}\alpha}_B:k=0,1,\ldots,\infty\}$. Hence, \textit{true hybrid qubit-qumode entanglement} is obtained. Furthermore, the form of the output state illustrates the effect of the thermal photon noise channel in a concrete way: On the one hand, the damping effect due to the beam splitter is clearly visible in terms of $\sqrt{\eta}$ in the states $\hat{a}^{\dagger^k}\ket{\pm\sqrt{\eta}\alpha}$. On the other hand, there is not only damping, but also thermal photon noise, which becomes manifest in the creation operators $\hat{a}^{\dagger^k}$ in the states $\hat{a}^{\dagger^k}\ket{\pm\sqrt{\eta}\alpha}$. Thermal photons ``leak into the system'' and are ``created'' in the damped coherent states. Finally, it is clearly visible how each term of $\hat{\rho}^E_{thermal}=\sum_{n=0}^\infty \rho_n^{th}\,\ket{n}^{E}\bra{n}$ results in the creation of at most $n$ noise photons in the coherent states. However, as descriptive as Eq. \eqref{eq:noisyQbQm3} is, it is inapplicable for further calculations. For example, for entanglement witnessing, moments may have to be computed \cite{SV1,SV2}. Unfortunately, in this case, making use of the state as written in Eq. \eqref{eq:noisyQbQm3}, intractable infinite sums are obtained whose convergence behavior is impossible to be worked out exactly. Of course, truncation at some value $n$ could be performed, which would certainly result in very accurate outcomes provided that value is large enough \cite{Killoran}. However, such a procedure is opposed to the actual intention of analyzing true hybrid entanglement, since a truncated state is not truly hybrid entangled anymore. This is a point which makes the investigation of true hybrid entanglement particularly challenging. Infinite sums or integrals emerge, which have to be calculated \textit{exactly}.

However, the transmitted state can be also written as
\begin{widetext}
\begin{equation}\label{eq:noisyQbQm4}
\begin{aligned}
\hat{\rho}'{}^{AB} & =({\mathbb 1}^A\otimes\$_{thermal}^B)\ket{\psi}^{AB}\bra{\psi} = \frac{1}{2\pi\braket{n_{th}}}\int_{\mathbb C} d^2\gamma\,\mathrm{e}^{-\frac{|\gamma|^2}{\braket{n_{th}}}}\biggl(\ket{0}^A\bra{0}\otimes\ket{\sqrt{\eta}\alpha-\sqrt{1-\eta}\gamma}^B\bra{\sqrt{\eta}\alpha-\sqrt{1-\eta}\gamma} \\
&+\qquad\qquad\ket{1}^A\bra{1}\otimes\ket{-\sqrt{\eta}\alpha-\sqrt{1-\eta}\gamma}^B\bra{-\sqrt{\eta}\alpha-\sqrt{1-\eta}\gamma} \\
&+\tilde{A}^{\alpha-\alpha\gamma}(\eta)\ket{0}^A\bra{1}\otimes\ket{\sqrt{\eta}\alpha-\sqrt{1-\eta}\gamma}^B\bra{-\sqrt{\eta}\alpha-\sqrt{1-\eta}\gamma} \\
&+\tilde{A}^{-\alpha\alpha\gamma}(\eta)\ket{1}^A\bra{0}\otimes\ket{-\sqrt{\eta}\alpha-\sqrt{1-\eta}\gamma}^B\bra{\sqrt{\eta}\alpha-\sqrt{1-\eta}\gamma}  \biggr), 
\end{aligned}
\end{equation}
\newpage
\end{widetext}
with
\begin{equation}
\begin{aligned}
\tilde{A}^{\alpha\beta\gamma}(\eta):&=\braket{\sqrt{1-\eta}\beta +\sqrt{\eta}\gamma|\sqrt{1-\eta}\alpha+\sqrt{\eta}\gamma} \\
&=\mathrm{exp}\Bigl[-\frac{1}{2}|\sqrt{1-\eta}\beta +\sqrt{\eta}\gamma|^2 \\
&\;\,-\frac{1}{2}|\sqrt{1-\eta}\alpha+\sqrt{\eta}\gamma|^2 \\
&\;\,+(\sqrt{1-\eta}\beta^{\ast} +\sqrt{\eta}\gamma^{\ast})(\sqrt{1-\eta}\alpha+\sqrt{\eta}\gamma)\Bigr],
\end{aligned}
\end{equation}
(see Appendix \ref{Apx:AuxCal} for the calculation). This form of the state $\hat{\rho}'{}^{AB}$ is not as insightful as that of Eq. \eqref{eq:noisyQbQm3}. However, it is mathematically much more convenient. Instead of infinite sums, in this form integrals occur which make the calculation of moments straightforward.

Now we want to derive entanglement witnesses for true hybrid entanglement. For this purpose, we exploit the determinant (now $\hat{a}$ and $\hat{a}^\dagger$ are the mode operators corresponding to system A and $\hat{b}$ and $\hat{b}^\dagger$ those of system B)
\begin{equation}\label{eqSdet}
s:=\begin{vmatrix} 
1 & \braket{\hat{a}^\dagger} & \braket{\hat{a}^\dagger\hat{b}}  \\
\braket{\hat{a}} & \braket{\hat{a}^\dagger\hat{a}} & \braket{\hat{a}^\dagger\hat{a}\hat{b}}  \\
\braket{\hat{a}\hat{b}^\dagger} & \braket{\hat{a}^\dagger\hat{a}\hat{b}^\dagger} & \braket{\hat{a}^\dagger\hat{a}\hat{b}^\dagger\hat{b}}
\end{vmatrix},
\end{equation}
which Shchukin and Vogel (SV) proposed and utilized in their work \cite{SV1,SV2}. However, we have to slightly adapt the SV approach to our case, since, instead of a fully CV system, we deal with a hybrid quantum system of a DV and a CV subsystem. There are two ways for this adaption.

The first way is to simply interpret the DV subsystem as living in a subspace of an infinite-dimensional Hilbert space. The DV qudit is interpreted as a CV system supported by ${\mathcal H}_\infty$ and being encoded in a Fock basis. However, it only makes use of a finite number of the basis vectors. Then the SV criteria can be applied just as usual.

The second approach is not to adapt the state, but to adapt the criteria. Assume the system to which the operators $\hat{a}$ and $\hat{a}^\dagger$ belong is $d$-dimensional. Then the orthonormal Hilbert-space basis vectors $\ket{m}$ can be written as column vectors with a ``$1$'' in row $m$ with $m=0,\ldots,d-1$:
\begin{equation}
\ket{m}_d=(0_{\,0}, \ldots , 0_{m-1} ,1_m ,0_{m+1},\ldots ,0_{d-1})^T.
\end{equation}
With this notation the new qudit mode operators $\hat{a}_d$ and $\hat{a}_d^\dagger$ can be defined as proper $d\times d$ matrices:
\begin{align}
\hat{a}_d &= \begin{pmatrix} 
0 & \sqrt{1} & \cdots & \cdots & 0 \\
\vdots & 0 & \sqrt{2} &   & \vdots \\
\vdots &   & 0 & \ddots  & \vdots \\
\vdots &  &  & \ddots & \sqrt{d-1} \\
0 & \cdots & \cdots & \cdots & 0
\end{pmatrix},\\
\hat{a}_d^\dagger &= \begin{pmatrix} 
0 & \cdots & \cdots & \cdots & 0 \\
\sqrt{1} & 0 &  &   & \vdots \\
\vdots &  \sqrt{2} & 0 &   & \vdots \\
\vdots &  & \ddots & \ddots & \vdots \\
0 & \cdots & \cdots & \sqrt{d-1} & 0
\end{pmatrix}.
\end{align}
They have the properties
\begin{align}
\hat{a}_d\ket{n} & = \sqrt{n}\ket{n-1}, \\
\hat{a}_d^\dagger\ket{n} & = (1-\delta_{n,d-1})\sqrt{n+1}\ket{n+1}, \\ \label{eq:AdMod1}
(\hat{a}_d)^d & =0, \\ \label{eq:AdMod2}
(\hat{a}_d^\dagger)^d & =0.
\end{align}
Furthermore, the commutator becomes
\begin{equation}
\begin{aligned}
\lbrack\hat{a}_d,\hat{a}_d^\dagger\rbrack&=\begin{pmatrix} 
1 &  &  &  \\
 & \ddots &  &  \\
 &  & 1 &  \\
 &  &  & -(d-1)
\end{pmatrix}_{d\times d} \\
&=\begin{pmatrix} 
{\mathbb 1}_{d-1} &    \\
 &  -(d-1)
\end{pmatrix}_{d\times d}.
\end{aligned}
\end{equation}
Since any $d$-dimensional qudit $\ket{\psi}_d$ can be written as $\ket{\psi}_d=g^\dagger(\hat{a}_d)\ket{0}_d$ for an appropriate operator function $\hat{g}=g(\hat{a}_d)$ and also the projection operator on $\ket{0}_d\bra{0}$ can still be expressed as $:e^{-\hat{a}_d^\dagger\hat{a}_d}:$, where $:\cdots:$ denotes normal ordering \cite{SV1}, the SV criteria can be also derived for hybrid systems with aid of these new operators. Structurally the same criteria are obtained, except that for certain $i_1,i_2,j_1,j_2$ the moments $M_{ij}(\hat{\rho})=\braket{\hat{a}_d^{\dagger^{i_2}}\hat{a}_d^{^{i_1}}\hat{a}_d^{\dagger^{j_1}}\hat{a}_d^{^{j_2}}\hat{b}^{\dagger^{j_4}}\hat{b}^{^{j_3}}\hat{b}^{\dagger^{i_3}}\hat{b}^{^{i_4}}}_{\hat{\rho}}$ are zero. Due to the properties \eqref{eq:AdMod1} and \eqref{eq:AdMod2} any combination of $i_1,i_2,j_1,j_2$ such that $(\hat{a}_d)^k$ or $(\hat{a}_d^\dagger)^k$ with $k\geq d$ occurs in the moments, nullifies them. In the end, in the matrix of moments the rows and columns corresponding to these combinations of $i_1,i_2,j_1,j_2$ are simply missing.

When applying the SV criteria it hardly makes a difference whether the first or the second approach to the adaption of the criteria is chosen. Only for moments involving terms like $(\hat{a}_d)^k(\hat{a}_d^\dagger)^l$ with $k,l\in\mathbb N$ the two approaches may yield different results. However, for the determinants used in this paper, both lead to the same result.  

For the calculation of the SV determinant $s$ in Eq. \eqref{eqSdet} for the state in Eq. \eqref{eq:noisyQbQm4},
\begin{equation}
s_{\hat{\rho}'{}^{AB}}(\alpha,\eta,\braket{n_{th}})=\frac{1-\eta}{4}\braket{n_{th}}\biggl(1-\frac{\mathrm{e}^{-4|\alpha|^2}}{2}\biggr)-\frac{\eta |\alpha|^2}{2}\mathrm{e}^{-4|\alpha|^2}
\end{equation}
is obtained. A graphical illustration for $\eta=\frac{2}{3}$ is shown in Fig. \ref{fig:ThQbQm1}.
\begin{figure}[ht]
\includegraphics[width=8.5cm]{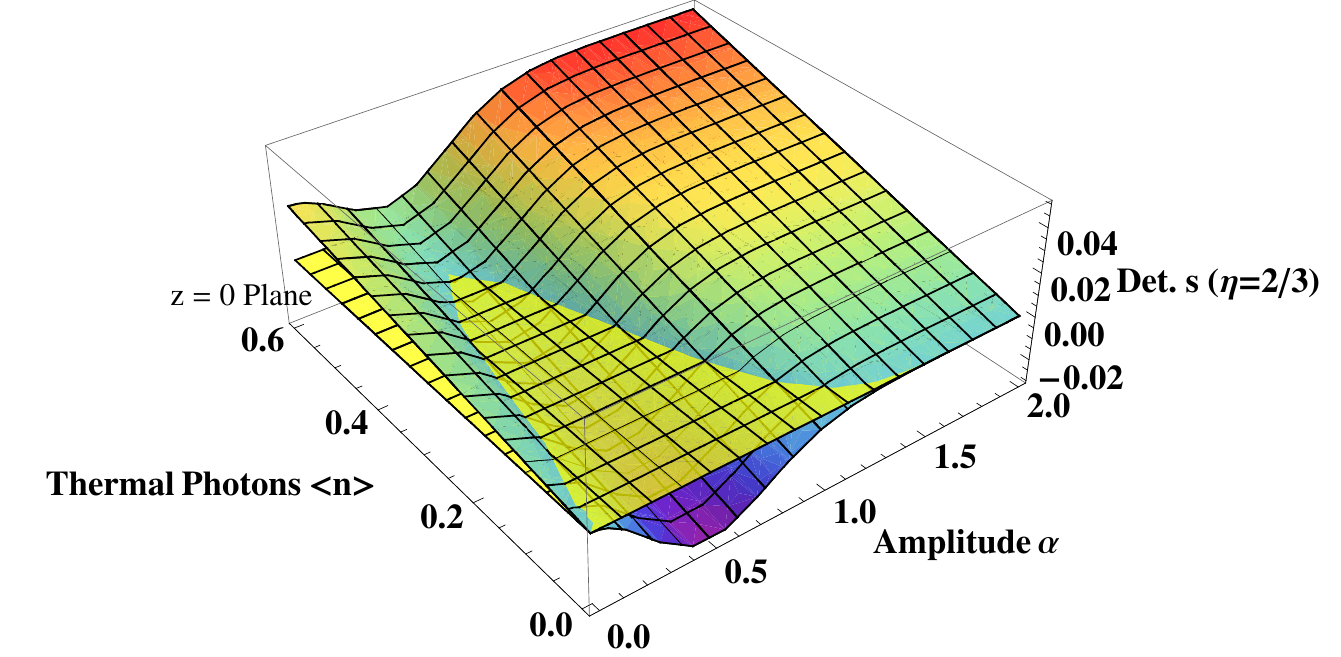}
\includegraphics[width=8.5cm]{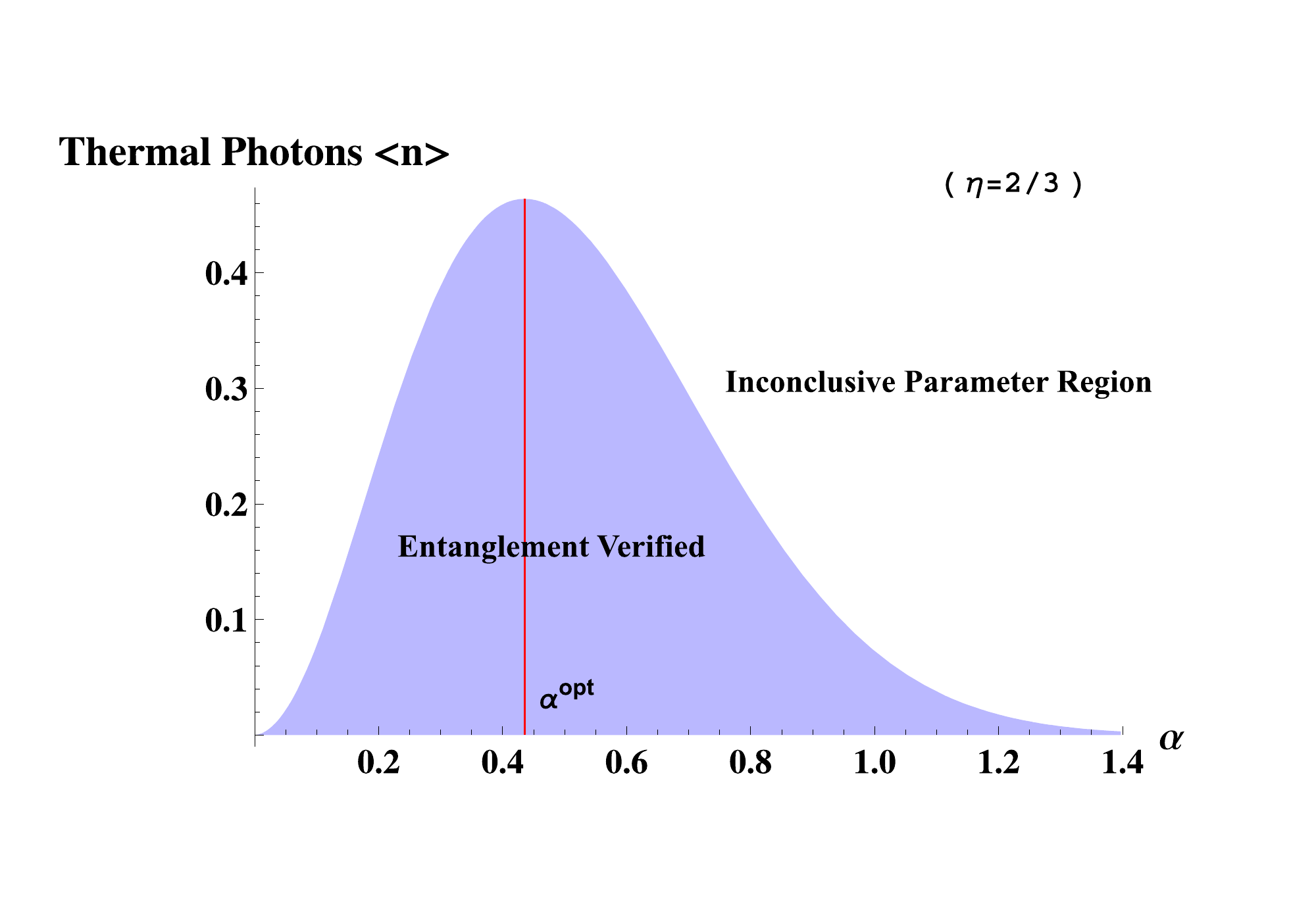}
\caption[SV determinant $s_{\hat{\rho}'{}^{AB}}(\alpha,\eta=\frac{2}{3},{\braket{n_{th}}})$ for the state $\hat{\rho}'{}^{AB}$ with $\eta=\frac{2}{3}$, and parameter regimes of entanglement witnessing.]{(Color online) The upper graph displays the SV determinant $s_{\hat{\rho}'{}^{AB}}(\alpha,\eta=\frac{2}{3},{\braket{n_{th}}})$ for the state $\hat{\rho}'{}^{AB}$ with $\eta=\frac{2}{3}$. Without loss of generality $\alpha\in\mathbb R$ has been assumed. On the lower diagram the two regimes are plotted. There can be clearly identified a region, where the determinant is below zero. Hence, true hybrid entanglement can be witnessed. Furthermore, there is a trade-off behavior which results in the existence of an optimal $\alpha^{opt}$. This $\alpha^{opt}$ corresponds to the most robust state $\hat{\rho}'{}^{AB}_{\alpha^{opt}}$ regarding entanglement witnessing for fixed $\eta=\frac{2}{3}$ and varying mean thermal photon number $\braket{n_{th}}$.}
\label{fig:ThQbQm1}
\end{figure}
It can be observed that there clearly is a parameter region in which entanglement can be detected. Also note the trade-off behavior and the optimal $\alpha^{opt}$ which corresponds to the most robust state $\hat{\rho}'{}^{AB}_{\alpha^{opt}}$ regarding entanglement witnessing for fixed $\eta=\frac{2}{3}$ and varying mean thermal photon number $\braket{n_{th}}$. The origin of this trade-off is based on the trade-off between too little initial entanglement before the channel for low $\alpha$ and very fragile entanglement for too high $\alpha$ ($\alpha \in \mathbb R$).

Furthermore, it is worth pointing out that the witnessed entanglement in this case is actually true hybrid entanglement. It can be concluded that the SV determinants provide a suitable tool for the detection of true hybrid entanglement. Note that we exploited determinants involving moments of $4^{th}$ order. There are also other approaches for witnessing entanglement in similar states, which involve lower-order moments \cite{Rigas,Haseler}. However, often it is necessary to consider moments of higher order to detect entanglement. For example, the detection of entanglement in the pure state ${\mathcal N}(\alpha)(\ket{\alpha,\alpha}-\ket{-\alpha,-\alpha})$ requires moments of $4^{th}$ order \cite{SV1}.  

When setting $s_{\hat{\rho}'{}^{AB}}(\alpha,\eta,\braket{n_{th}})<0$ and solving this inequality for $\braket{n_{th}}$,
\begin{equation}\label{eq:ThQbQmCut}
\braket{n_{th}}<\frac{4\eta|\alpha|^2}{(1-\eta)(2\mathrm{e}^{4|\alpha|^2}-1)}
\end{equation}
is obtained \footnote[2]{Comparing Eq. \eqref{eq:ThQbQmCut} with the known entanglement-breaking condition for a thermal noise channel \cite{Holevo}, $\braket{n_{th}}\geq\frac{\eta}{1-\eta}$, reveals that our witness would indeed never wrongly detect entanglement, while at the same time, it may not detect all entangled states if $\braket{n_{th}}<\frac{\eta}{1-\eta}$.}. For parameters $(\alpha,\eta,\braket{n_{th}})$ satisfying this inequality, entanglement is detected. Furthermore, the inequality can be used to define a surface. The parameters $(\alpha,\eta,\braket{n_{th}})_{ent}$ for which entanglement is verified lie below this surface (see Fig. \ref{fig:ThQbQm2}).
\begin{figure}[ht]
\begin{center}
\includegraphics[width=8.5cm]{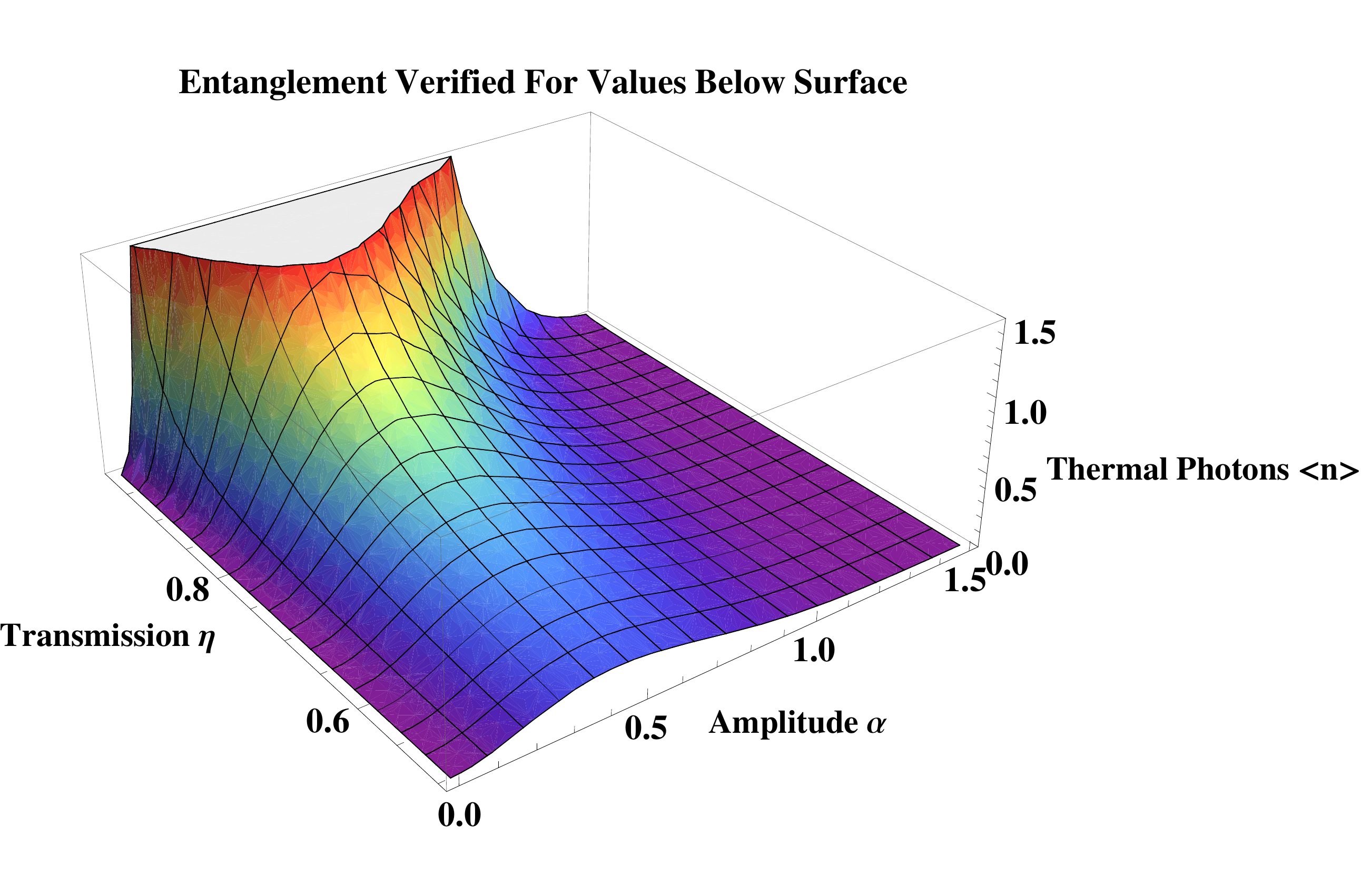}
\caption[Surface defined by equation \eqref{eq:ThQbQmCut} describing entanglement witnessing in the truly hybrid entangled state $\hat{\rho}'{}^{AB}$ corresponding to equation \eqref{eq:noisyQbQm4}.]{(Color online) Surface defined by Eq. \eqref{eq:ThQbQmCut} ($\alpha\in\mathbb R$). For parameter triples $(\alpha,\eta,\braket{n_{th}})$ lying below it, entanglement is witnessed. Once again, the trade-off behavior can be recognized and an optimal $\alpha^{opt}$ for which entanglement can be detected in the presence of the strongest possible noise.}
\label{fig:ThQbQm2}
\end{center}
\end{figure}
However, it is rather cumbersome to read off exact parameters in such a 3D plot. Hence, regions of successful entanglement detection are plotted in Fig. \ref{fig:ThQbQm3} for different values of $\eta$. 
\begin{figure}[ht]
\begin{center}
\includegraphics[width=8.5cm]{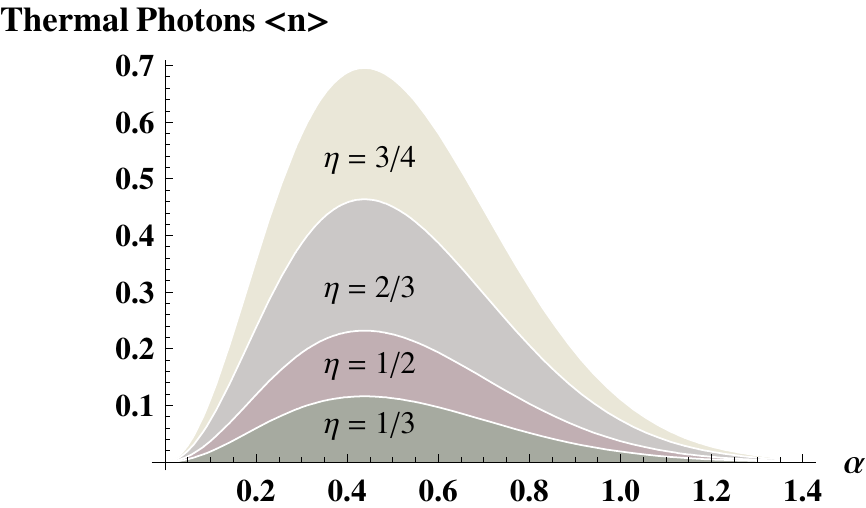}
\caption[Regions of entanglement detection in the truly hybrid entangled state $\hat{\rho}'{}^{AB}$ corresponding to Eq. \eqref{eq:noisyQbQm4} for different values of transmissivity $\eta$.]{(Color online) Regions of entanglement detection for different values of transmissivity $\eta$. Note that the optimal $\alpha^{opt}$ has a fixed value ($\alpha\in\mathbb R$).}
\label{fig:ThQbQm3}
\end{center}
\end{figure}
As expected, the higher the transmissivity $\eta$ the greater the parameter regions of entanglement detection. Furthermore, we can infer from Eq. \eqref{eq:ThQbQmCut} that the optimal $\alpha^{opt}$ does not depend on $\eta$. This is also recognizable in Fig. \ref{fig:ThQbQm3}. We find $\alpha^{opt}\approx 0.44$. It is actually quite remarkable that the optimal amplitude $\alpha^{opt}$ regarding entanglement witnessing with the SV determinant $s$ does not depend on the channel parameters at all; $\alpha^{opt}$ is determined solely by the choice of the SV determinant. 

\begin{figure}[ht]
\begin{center}
\includegraphics[width=8.5cm]{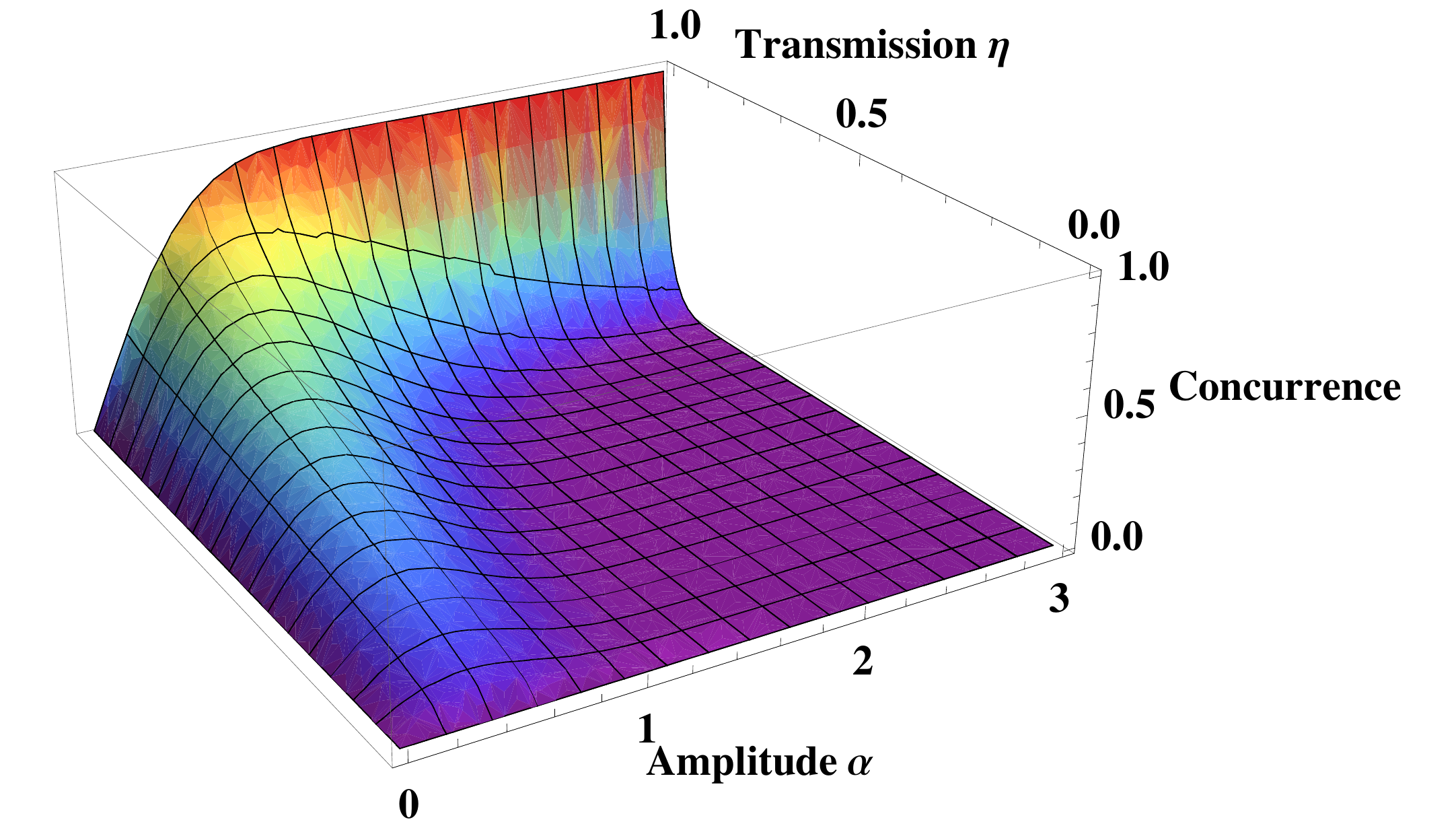}
\caption[Concurrence of the state $\hat{\rho}^{AB}$, as a function of transmissivity $\eta$ and amplitude $\alpha$.]{(Color online) Concurrence of the state $\hat{\rho}'{}^{AB}$ for $\braket{n_{th}}=0$, as a function of transmissivity $\eta$ and amplitude $\alpha$. On the one hand, the higher $\alpha$ the greater the initial entanglement, but also the more sensitive the state is to photon loss. On the other hand, for low $\alpha$ there is only little initial entanglement. However, the state is more robust against losses. Hence, there is a trade-off behavior, and an optimal $\alpha^{opt}_\eta$ depending on $\eta$ exists for which the output entanglement is maximal.}
\label{fig:CADampCh}
\end{center}
\end{figure}
Compare this to Fig. \ref{fig:CADampCh}, which shows the concurrence \cite{Wootters} of the state when setting $\braket{n_{th}}=0$, which is given by 
\begin{equation}\label{BlaCon}
\begin{aligned}
C(\hat{\rho}'{}^{AB})=&\frac{1}{2}\sqrt{1-\mathrm{e}^{-4\eta|\alpha|^2}}\Bigl(\sqrt{1+3\mathrm{e}^{-4(1-\eta)|\alpha|^2}} \\
&-\sqrt{1-\mathrm{e}^{-4(1-\eta)|\alpha|^2}}\Bigr).
\end{aligned}
\end{equation}
There, the optimal $\tilde{\alpha}^{opt}(\eta)$ does depend on the transmissivity $\eta$: The greater the transmissivity, the higher $\tilde{\alpha}^{opt}(\eta)$. This is quite remarkable. It can be inferred that the ability to detect entanglement in $\hat{\rho}'{}^{AB}$, depending only on the choice of the SV determinant, does not behave in the same way as the entanglement of $\hat{\rho}'{}^{AB}$ itself, which of course depends on the thermal channel's parameters $\eta$ and $\braket{n_{th}}$.

Note that in this case, the state of Eqs. \eqref{eq:noisyQbQm1} and \eqref{eq:noisyQbQm2} actually becomes an effectively DV hybrid entangled state, when setting the mean thermal photon number of the channel to zero. The state is then subject only to amplitude damping, resulting in
\begin{equation}\label{Eq0thermCase}
\begin{aligned}
\hat{\rho}'{}^{AB}_{\braket{n_{th}}=0} &=\frac{1}{2}\Bigl(\ket{0}^A\bra{0}\otimes\ket{\sqrt{\eta}\alpha}^B\bra{\sqrt{\eta}\alpha} \\
&+ \ket{1}^A\bra{1}\otimes\ket{-\sqrt{\eta}\alpha}^B\bra{-\sqrt{\eta}\alpha} \\
&+\mathrm{e}^{-2(1-\eta)|\alpha|^2}\, \ket{0}^A\bra{1}\otimes\ket{\sqrt{\eta}\alpha}^B\bra{-\sqrt{\eta}\alpha} \\
&+\mathrm{e}^{-2(1-\eta)|\alpha|^2}\, \ket{1}^A\bra{0}\otimes\ket{-\sqrt{\eta}\alpha}^B\bra{\sqrt{\eta}\alpha}\Bigr),
\end{aligned}
\end{equation}
which is mixed but contains only two qumode states $\ket{\pm\sqrt{\eta}\alpha}^B$. Hence, we can perform an inverse Gram-Schmidt process and describe the state by a proper effective qubit-qubit density matrix:
\begin{equation}
\hat{\rho}'{}^{AB}_{\braket{n_{th}}=0} = \frac{1}{2}\begin{pmatrix}
1 & 0 & \lambda\kappa & \kappa\sqrt{1-\lambda^2} \\
0 & 0 & 0 & 0 \\
\lambda\kappa & 0 & \lambda^2 & \lambda\sqrt{1-\lambda^2} \\
\kappa\sqrt{1-\lambda^2} & 0 & \lambda\sqrt{1-\lambda^2} & 1-\lambda^2
\end{pmatrix},
\end{equation}
with $\kappa:=\mathrm{e}^{-2(1-\eta)|\alpha|^2}$ and $\lambda:=\braket{-\sqrt{\eta}\alpha|\sqrt{\eta}\alpha}=\mathrm{e}^{-2\eta|\alpha|^2}$. The concurrence of this state is given in Eq. \eqref{BlaCon} and plotted in Fig. \ref{fig:CADampCh} for $\alpha \in \mathbb R$.

If, besides setting $\braket{n_{th}}=0$, we also set the transmissivity of the channel to one, the channel is completely canceled and we obtain a pure bipartite hybrid entangled state:
\begin{equation}\label{eq:noisyQbQm1norm}
\ket{\psi}^{AB}=\frac{1}{\sqrt{2}}\Bigl(\ket{0}^A\ket{\alpha}^B+\ket{1}^A\ket{-\alpha}^B\Bigr).
\end{equation}
Hence, in this physical example, we can access all three classes of bipartite hybrid entanglement and even switch between them using the transmissivity $\eta$ and the mean thermal photon number $\braket{n_{th}}$.

\subsection{Yet another truly hybrid entangled state}\label{subsec:OtherTrueHE}
In order to demonstrate that we can easily construct other quantum states featuring true hybrid entanglement, we shall present another truly hybrid entangled state and its entanglement detection:
\begin{equation}\label{eq:trueHE2}
\begin{aligned}
\hat{\rho}^{AB}&=\sum_{n=1}^{\infty}p_n\ket{\psi_n}^{AB}\bra{\psi_n}, \\
\ket{\psi_n}^{AB}&=\frac{1}{\sqrt{2}}\Bigl(\ket{0}^A\ket{\sqrt{n}\alpha}^B+\ket{1}^A\ket{-\sqrt{n}\alpha}^B\Bigr), \\
p_n&=\frac{1-x}{x}x^n\,,\qquad 0<x<1\,,\qquad \alpha\in\mathbb R.
\end{aligned}
\end{equation}
As this state contains an infinite set of qumode states $\{\ket{\pm\sqrt{n}\alpha}^B:n=1,2,\ldots,\infty\}$, it is clearly truly hybrid entangled, exhibiting (true) hybrid entanglement between a qumode and a qubit. 

Now, we shall witness entanglement in this state with aid of the SV determinant $s$ of Eq. \eqref{eqSdet}. For the calculation of the moments, the identities
\begin{align}
\sum_{n=1}^\infty x^n &= \frac{x}{1-x}, \\ 
\sum_{n=1}^\infty n\,x^n &= \frac{x}{(1-x)^2},
\end{align}
are exploited, and we obtain
\begin{equation}\label{eq:ArtiS}
\begin{aligned}
s(x,\alpha)&=\frac{1}{8}\biggl[\frac{2\alpha^2}{1-x}-\Bigl(\frac{\alpha(1-x)}{x}\sum_{n=1}^\infty\sqrt{n}(x\mathrm{e}^{-2\alpha^2})^n\Bigr)^2 \\
&-2\Bigl(\frac{\mathrm{e}^{-2\alpha^2}(1-x)}{1-x\mathrm{e}^{-2\alpha^2}}\Bigr)\Bigl(\frac{\alpha(1-x)}{x}\sum_{n=1}^\infty\sqrt{n}(x\mathrm{e}^{-2\alpha^2})^n\Bigr) \\
&\times\Bigl(\frac{\alpha(1-x)}{x}\sum_{n=1}^\infty\sqrt{n}x^n\Bigr)-2\Bigl(\frac{\alpha(1-x)}{x}\sum_{n=1}^\infty\sqrt{n}x^n\Bigr)^2 \\
&-\frac{\alpha^2}{1-x}\Bigl(\frac{\mathrm{e}^{-2\alpha^2}(1-x)}{1-x\mathrm{e}^{-2\alpha^2}}\Bigr)^2\biggr].
\end{aligned}
\end{equation}
It is clear that
\begin{equation}
\sum_{n=1}^\infty x^n<\sum_{n=1}^\infty \sqrt{n}\,x^n<\sum_{n=1}^\infty n\,x^n,
\end{equation}
for $0<x<1$. Hence,
\begin{align}
\frac{x}{1-x}< & \sum_{n=1}^\infty \sqrt{n}\,x^n<\frac{x}{(1-x)^2}, \\
\frac{x\,\mathrm{e}^{-2\alpha^2}}{1-x\,\mathrm{e}^{-2\alpha^2}}< & \sum_{n=1}^\infty \sqrt{n}\,(x\,\mathrm{e}^{-2\alpha^2})^n < \frac{x\,\mathrm{e}^{-2\alpha^2}}{(1-x\,\mathrm{e}^{-2\alpha^2})^2}.
\end{align}
Inserting the lower bounds into the sums of Eq. \eqref{eq:ArtiS} yields
\begin{equation}
s'(x,\alpha)=\frac{\alpha^2}{8}\Bigl[\frac{2x}{1-x}-\Bigl(\frac{1-x}{1-x\mathrm{e}^{-2\alpha^2}}\Bigr)^2\mathrm{e}^{-4\alpha^2}\Bigl(3+\frac{1}{1-x}\Bigr)\Bigr].
\end{equation}
Since $s'(x,\alpha)>s(x,\alpha)\,\forall\,\alpha\in{\mathbb R},\,x\in]0,1[$, from $s'(x,\alpha)<0$ follows $s(x,\alpha)<0$. Therefore, $s'(x,\alpha)<0$ is a sufficient criterion for entanglement detection. It is plotted in Fig. \ref{fig:artiS}.
\begin{figure}[ht]
\includegraphics[width=8.5cm]{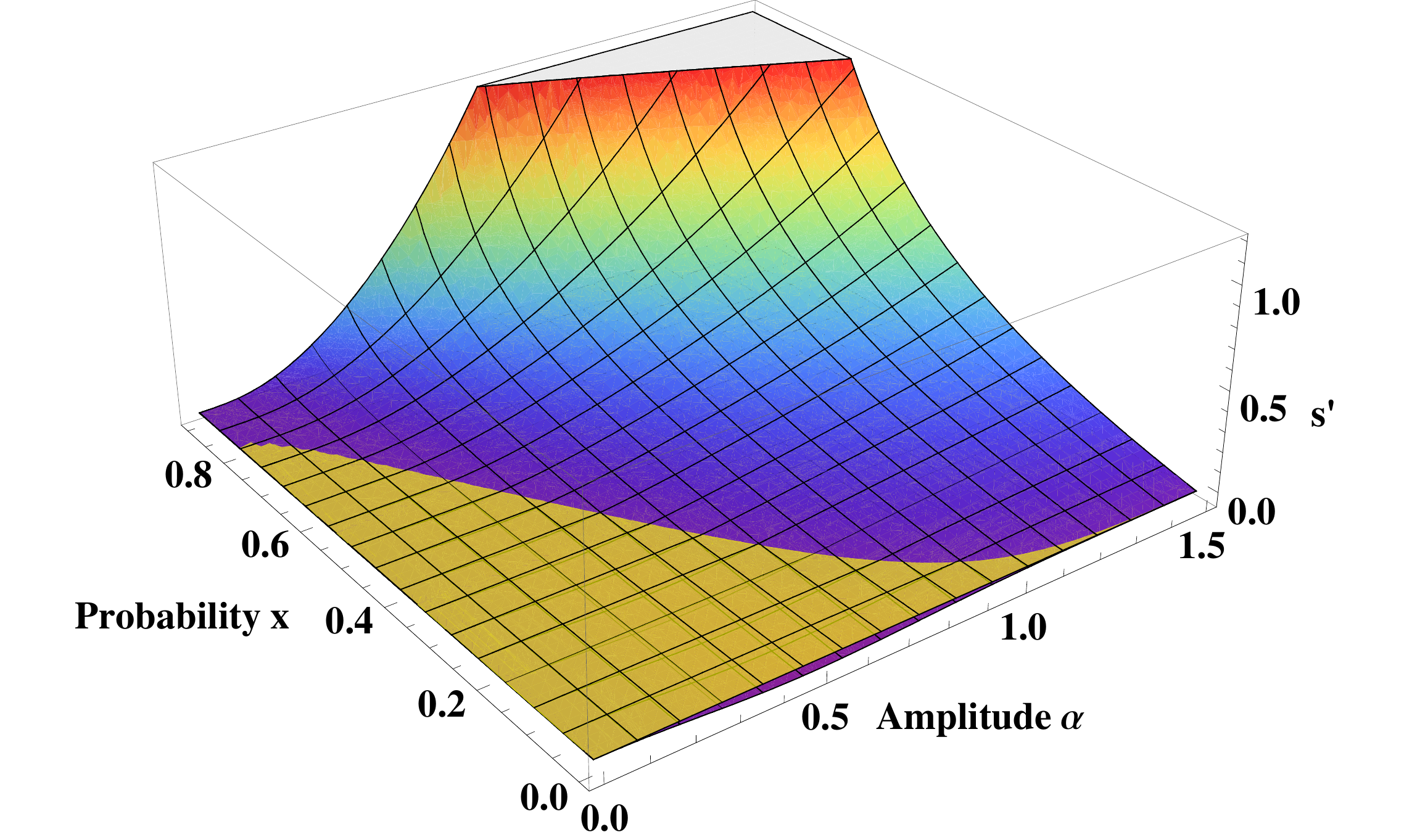}
\includegraphics[width=8.5cm]{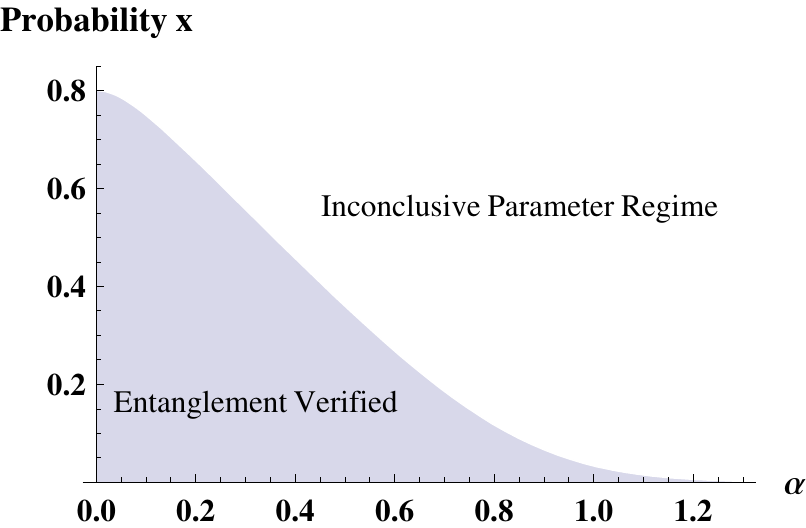}
\caption[Entanglement witnessing in the state corresponding to equation \eqref{eq:trueHE2} via $s'_1(x,\alpha)$.]{(Color online) The upper graph shows $s'(x,\alpha)$ which is derived from the original SV determinant $s(x,\alpha)$. %Note that for the plot the prefactor $\frac{\alpha^2}{8}$ has been omitted, which does not affect the witnessing region. 
The yellow plane denotes zero. The lower diagram presents the cutting line between $s'(x,\alpha)$ and the zero plane. It displays the witnessing region. We also want to point out that for $\alpha=0$ the state is of course not entangled and $s'(x,\alpha=0)$ becomes zero.}
\label{fig:artiS}
\end{figure}
The graphics show that entanglement can be verified for sufficiently small $x$ and $\alpha$. This can be understood in the following way. On the one hand, small $x$ corresponds to little mixing between the pure states $\ket{\psi_n}^{AB}$, which themselves are highly entangled, depending on the amplitude $\sqrt{n}\alpha$. For pure states of the form $\ket{\psi_n}^{AB}$, entanglement witnessing can be performed perfectly with the SV determinant $s$. On the other hand, we have already seen in previous calculations (see Figs. \ref{fig:ThQbQm1} and \ref{fig:ThQbQm3}) that entanglement detection via $s$ fails for large $\alpha$.

To conclude, a second example for a truly hybrid entangled state has been presented whose entanglement can again be verified with aid of the SV criteria.

\section{Multipartite hybrid entanglement}\label{sec:4}
In this section, investigations regarding multipartite hybrid entanglement are presented. In general, $N$-partite hybrid entangled quantum systems live in Hilbert spaces of the form ${\mathcal H}_{d_1}\otimes\ldots\otimes{\mathcal H}_{d_N}$, where some $d_i$ are finite and some infinite. Considering tripartite hybrid entanglement there are two cases: Either the Hilbert space looks like ${\mathcal H}_{d_1}\otimes{\mathcal H}_{d_2}\otimes{\mathcal H}_\infty$ or ${\mathcal H}_{d_1}\otimes{\mathcal H}_\infty\otimes{\mathcal H}_\infty$, with finite $d_1,d_2$. 

Consider the first case, where only one subsystem is CV. A general hybrid entangled pure state in ${\mathcal H}_{d_1}^A\otimes{\mathcal H}_{d_2}^B\otimes{\mathcal H}^C_\infty$ may be defined as
\begin{equation}
\ket{\psi}^{ABC}=\sum_{i,j=1}^{d_1,d_2}c_{ij}\ket{e_i}^A\ket{e_j}^B\ket{\psi_{ij}}^C\,,\qquad\sum_{i,j}^{d_1,d_2}|c_{ij}|^2=1,
\end{equation}
where $\ket{\psi_{ij}}^C$ represent some qumode states. As there are at most $d_1\times d_2$ such qumode states, a Gram-Schmidt process can be executed to write the state as a multipartite, effectively DV hybrid entangled state. Therefore, if all $\ket{\psi_{ij}}^C$ are linearly independent the state effectively lives in a Hilbert space of the form ${\mathcal H}_{d_1}^A\otimes{\mathcal H}_{d_2}^B\otimes{\mathcal H}^C_{d_1\times d_2}$. 

In the second case, where the initial Hilbert space looks like ${\mathcal H}_{d}^A\otimes{\mathcal H}_\infty^B\otimes{\mathcal H}_\infty^C$, a general hybrid entangled pure state is
\begin{equation}\label{MultiHy2}
\ket{\psi}^{ABC}=\sum_{i,j=1}^{d,\infty}c_{ij}\ket{e_i}^A\ket{\phi_{ij}}^B\ket{\psi_{ij}}^C,
\end{equation}
where both $\ket{\phi_{ij}}^B$ and $\ket{\psi_{ij}}^C$ represent qumode states, and the $c_{ij}$ are chosen such that $\braket{\psi|\psi}=1$. Obviously, in this case an infinite number of qumode states is present for an infinite number of $c_{ij}\neq0$. No Gram-Schmidt process can be performed. Hence, these states show \textit{multipartite true hybrid entanglement}.

To conclude, in the multipartite regime, already for pure states there are characteristic differences between the different possible configurations. Furthermore, in contrast to the bipartite setting, there is also \textit{pure} true hybrid entanglement. When dealing with two parties, true hybrid entanglement only occurs for specific types of mixed states, as described in the previous sections. This can now be generalized. True hybrid entanglement can be obtained in two ways: Either the system is mixed with an infinite number of mix terms; then it is sufficient that only one subsystem is CV. Or two or more subsystems are CV; then the state is not even required to be mixed. However, every hybrid entangled mixed state with a finite number of mix terms, which contains only one CV subsystem is effectively DV. An $N$-partite mixed state of this type in ${\mathcal H}_{d_1}\otimes\ldots\otimes{\mathcal H}_{d_{N-1}}\otimes{\mathcal H}_{\infty}$ can be always described in a Hilbert space of the form ${\mathcal H}_{d_1}\otimes\ldots\otimes{\mathcal H}_{d_{N-1}}\otimes{\mathcal H}_{\Xi}$, where
\begin{equation}
\Xi=M\prod_{i=1}^{N-1} d_i,
\end{equation}
and $M$ is the number of mix terms. 

As an example, we want to discuss an explicit tripartite hybrid entangled state, supported by a Hilbert space of the form ${\mathcal H}_{d}\otimes{\mathcal H}_\infty\otimes{\mathcal H}_\infty$, with finite $d=2$:
\begin{equation}\label{MultiHyExamp}
\ket{\tilde{\psi}}^{ABC}=\frac{1}{\sqrt{2}}\biggl(\ket{e_0}^{A}\ket{\phi_0}^{B}\ket{\psi_0}^{C}+\ket{e_1}^{A}\ket{\phi_1}^{B}\ket{\psi_1}^{C}\biggr).
\end{equation}
In this example, the special case occurs that there does not exist an infinite number of $c_{ij}\neq0$. In terms of Eq. \eqref{MultiHy2}, $c_{ij}=\frac{1}{\sqrt{2}}$ for $i=1,2$ and $j=1$, and $c_{ij}=0\;\forall \;j>1$. Hence, there is no true hybrid entanglement and the state can be described by DV methods. 

Defining ${\mathcal Q}_{\phi}:=\braket{\phi_0|\phi_1}$ and ${\mathcal Q}_{\psi}:=\braket{\psi_0|\psi_1}$, an inverse Gram-Schmidt process yields
\begin{align}
\ket{\phi_0}^{B}&=\ket{e_0}^{B}, \\
\ket{\phi_1}^{B}&={\mathcal Q}_{\phi}\ket{e_0}^{B} + \sqrt{1-|{\mathcal Q}_{\phi}|^2}\ket{e_1}^{B}, \\
\ket{\psi_0}^{C}&=\ket{e_0}^{C}, \\
\ket{\psi_1}^{C}&={\mathcal Q}_{\psi}\ket{e_0}^{C} + \sqrt{1-|{\mathcal Q}_{\psi}|^2}\ket{e_1}^{C},
\end{align}
and therefore,
\begin{equation}
\begin{aligned}
\ket{\tilde{\psi}}^{ABC}&=\frac{1}{\sqrt{2}}\biggl(\ket{e_0}^{A}\ket{e_0}^{B}\ket{e_0}^{C}+{\mathcal Q}_{\phi}{\mathcal Q}_{\psi}\ket{e_1}^{A}\ket{e_0}^{B}\ket{e_0}^{C} \\
& +\sqrt{(1-|{\mathcal Q}_{\phi}|^2)(1-|{\mathcal Q}_{\psi}|^2)}\ket{e_1}^{A}\ket{e_1}^{B}\ket{e_1}^{C} \\
& +{\mathcal Q}_{\psi}\sqrt{1-|{\mathcal Q}_{\phi}|^2}\ket{e_1}^{A}\ket{e_1}^{B}\ket{e_0}^{C} \\
& +{\mathcal Q}_{\phi}\sqrt{1-|{\mathcal Q}_{\psi}|^2}\ket{e_1}^{A}\ket{e_0}^{B}\ket{e_1}^{C}\biggr).
\end{aligned}
\end{equation}
The state is effectively a three-qubit state in ${\mathcal H}_{2}^A\otimes{\mathcal H}_{2}^B\otimes{\mathcal H}^C_2$. So, we can analyze the state with regard to its bipartite entanglement and its residual, GHZ-like, genuine tripartite entanglement $\tau_{res}$ \cite{CKW}. The relevant squared concurrences are
\begin{align}
C^2(A|B)&=|{\mathcal Q}_{\psi}|^2(1-|{\mathcal Q}_{\phi}|^2) \notag \\
& =|\braket{\psi_0|\psi_1}|^2(1-|\braket{\phi_0|\phi_1}|^2) \\
C^2(A|C)&=|{\mathcal Q}_{\phi}|^2(1-|{\mathcal Q}_{\psi}|^2) \notag \\
& =|\braket{\phi_0|\phi_1}|^2(1-|\braket{\psi_0|\psi_1}|^2),  \\
C^2(B|C)&=0, \\ 
C^2(A|BC)&=1-|{\mathcal Q}_{\phi}|^2|{\mathcal Q}_{\psi}|^2 \notag \\
& =1-|\braket{\phi_0|\phi_1}|^2|\braket{\psi_0|\psi_1}|^2.
\end{align}
Hence,
\begin{equation}
\begin{aligned}
\tau_{res}&=(1-|{\mathcal Q}_{\phi}|^2)(1-|{\mathcal Q}_{\psi}|^2) \\
&=(1-|\braket{\phi_0|\phi_1}|^2)(1-|\braket{\psi_0|\psi_1}|^2),
\end{aligned}
\end{equation}
which is plotted in Fig. \ref{fig:resEnt}. 
\begin{figure}[ht]
\begin{center}
\includegraphics[width=8.5cm]{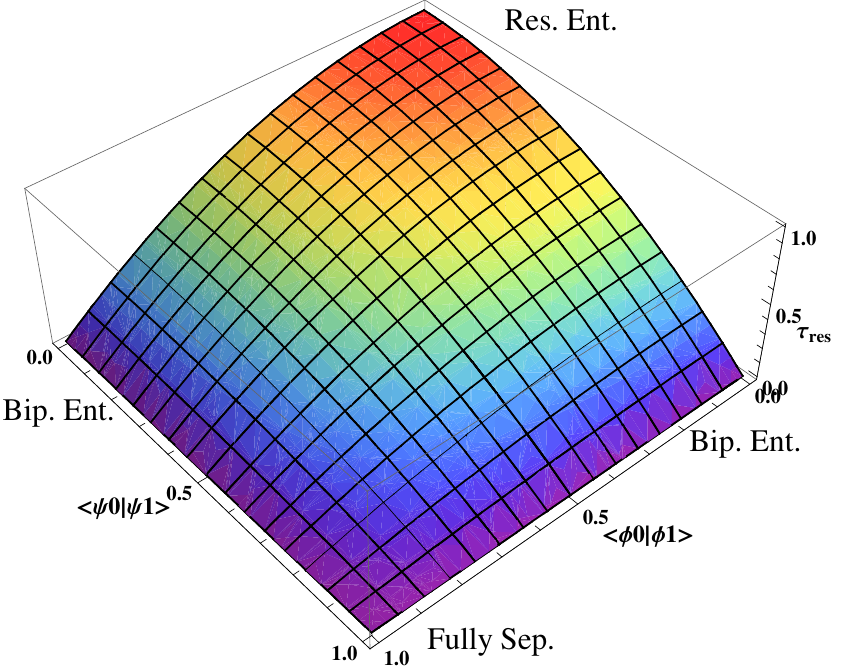}
\caption[The residual entanglement of the state $\ket{\psi}^{ABC}$ as a function of the overlaps $\braket{\phi_0|\phi_1}$ and $\braket{\psi_0|\psi_1}$.]{(Color online) The residual entanglement of the state $\ket{\tilde{\psi}}^{ABC}$ as a function of the overlaps $\braket{\phi_0|\phi_1}$ and $\braket{\psi_0|\psi_1}$, which are assumed to be real without loss of generality. In the upper corner, for $\tau_{res}=1$, the state is maximally GHZ-like entangled. It shows no bipartite entanglement. In the left and the right corners, the residual entanglement is zero. Nevertheless, the state is maximally entangled, but just bipartite entangled. Finally, in the lower corner, the state is fully separable. For either $\braket{\phi_0|\phi_1}$=0 or $\braket{\psi_0|\psi_1}$=0, the state is always 1-ebit entangled, while for both $\braket{\phi_0|\phi_1}\neq0$ and $\braket{\psi_0|\psi_1}\neq0$ the state has less than 1-ebit entanglement.}
\label{fig:resEnt}
\end{center}
\end{figure}
Furthermore, the total entanglement $C^2_{total}$ is given by
\begin{equation}
\begin{aligned}
C^2_{total}&=C^2(A|B)+C^2(A|C)+C^2(B|C)+\tau_{res} \\
&=1-|\braket{\phi_0|\phi_1}|^2|\braket{\psi_0|\psi_1}|^2.
\end{aligned}
\end{equation}
These results provide a basis for an interesting discussion. Figure \ref{fig:resEnt} looks quite unspectacular. However, it contains very useful information in combination with the other results. First, it can be seen that for both overlaps $\braket{\phi_0|\phi_1}$ and $\braket{\psi_0|\psi_1}$ being zero, the state is exactly the GHZ state, which shows only GHZ-like tripartite entanglement but no bipartite entanglement. In the case that only one of the overlaps is zero, the state is always 1-ebit entangled, while it can be tuned between GHZ-like entanglement and common bipartite entanglement. If the overlap which is not zero becomes one, maximal Bell-state-like entanglement occurs. Finally, if none of the overlaps is zero the state is not 1-ebit entangled anymore and it corresponds to a mixture of bipartite entanglement and GHZ-like tripartite entanglement. If both overlaps are one, the state becomes simply fully separable, showing no entanglement at all. All this is basically illustrated in Fig. \ref{fig:resEnt}. The upper corner of the graph corresponds to a GHZ state, while the left and the right corners represent states which are maximally bipartite entangled. Every point in the graph which is not at the upper left or upper right edge corresponds to a less than 1-ebit entangled state, while states lying on these edges are 1-ebit entangled. States on the lower edges of the graph are only bipartite entangled, while all other states show also GHZ-like tripartite entanglement. Finally, the lower corner corresponds to a fully separable state.

As a result, just by tuning simple overlaps between the involved qumode states, it is possible to gradually switch between these different entanglement scenarios. Such tuning of overlaps is a relatively easy task when the qumode states are simply realized by ordinary coherent states. Then only the amplitudes have to be adjusted. However, the experimental preparation of the overall entangled tripartite state $\ket{\tilde{\psi}}^{ABC}$ may cause difficulties. Nevertheless, this idea of tuning between various entanglement configurations by modification of the overlaps of the participating qumode states may be a scheme which could possibly be experimentally realized in the future.

\section{Summary and Conclusion}\label{sec:5}
In this paper, first, we proved that hybrid entanglement is necessarily non-Gaussian. Then, utilizing an inverse Gram-Schmidt process, we presented a classification scheme regarding hybrid entanglement which distinguishes between effective DV hybrid entanglement and true hybrid entanglement. With aid of this characterization one can always find out whether a given hybrid entangled state can be analyzed by means of DV methods or not. To illustrate this framework, a few exemplary states have been discussed, and witnessing of true hybrid entanglement has been demonstrated exploiting the well-known SV inseparability criteria. Finally, we examined multipartite hybrid entanglement, especially the tripartite case, and briefly outlined some differences compared to the bipartite regime. As an example, a typical tripartite hybrid entangled state was investigated.

To conclude, we want to point out that, although many ideas for hybrid protocols and schemes have appeared recently, no thorough classification of hybrid entanglement from a more formal point of view has been proposed so far. The present work is supposed to fill this gap. Nevertheless, there are many open questions: For example, entanglement quantification of true hybrid entanglement remains an unsolved problem in general. Especially in the multipartite case, various questions remain. For example, what can happen when we consider more parties than just two or three?

\subsection*{ACKNOWLEDGMENT}
Support from the Emmy Noether Program of the Deutsche Forschungsgemeinschaft is gratefully acknowledged. In addition, we thank the BMBF in Germany for support through the QuOReP program.

\begin{appendix}

\section{PROOF OF LEMMA \ref{lem:DVnonG}}\label{Apx:PrfLemma}
Lemma \ref{lem:DVnonG} states that any single-partite $d$-dimensional quantum state with finite $d$ and $d\geq2$ is non-Gaussian.
\begin{proof}
Consider such a general $d$-dimensional state in a pure-state decomposition and make use of the Fock basis of a $d$-dimensional subspace of the Fock space:
\begin{equation}\label{eq:Lemmaqudit}
\begin{aligned}
\hat{\rho} &=\sum_{n=1}^N p_n\ket{\psi_n}\bra{\psi_n}\,,\quad p_n>0\;\forall \;n\,;\;\sum_n^N \;p_n=1, \\
\ket{\psi_n} &=\sum_{m=0}^{d-1} c_{nm}\ket{m}\,,\qquad\; c_{nm}\in\mathbb C\;\forall \;n,m\,;\;\sum_m^{d-1} \;|c_{nm}|^2=1.
\end{aligned} 
\end{equation}
Non-Gaussianity of the state corresponds to non-Gaussianity of the characteristic function. This is equivalent to non-Gaussianity of the Wigner function, as (non-)Gaussian functions stay (non-)Gaussian under Fourier transformation \cite{Bronstein}. Hence, consider the Wigner function of $\hat{\rho}$,
\begin{equation}
\begin{aligned}
W(x,p) & :=\frac{1}{\pi}\int\limits_{-\infty}^\infty dy\,\mathrm{e}^{2ipy}\braket{x-y|\hat{\rho}|x+y} \\
& =\sum_{n=1}^N\frac{p_n}{\pi}\int\limits_{-\infty}^\infty dy\,\mathrm{e}^{2ipy} \sum_{k,l=0}^{d-1}c_{nl}c_{nk}^\ast\,\psi_l(x-y)\psi_k^\ast(x+y).
\end{aligned} 
\end{equation}
The position wavefunction of a Fock state $\ket{n}$ is given by
\begin{equation}
\psi_n(x)=\braket{x|n}=\frac{H_n(x)}{\sqrt{2^nn!\sqrt{\pi}}}\mathrm{e}^{-\frac{x^2}{2}},
\end{equation}
where $H_n(x)$ are the \textit{Hermite polynomials} \cite{Leon,Bronstein}. Therefore,
\begin{equation}
\begin{aligned}
W(x,p) & =\sum_{n=1}^N\frac{p_n}{\pi}\int\limits_{-\infty}^\infty dy\,\mathrm{e}^{2ipy}\,\mathrm{e}^{-\frac{(x-y)^2}{2}}\mathrm{e}^{-\frac{(x+y)^2}{2}} \\
& \times\sum_{k,l=0}^{d-1}c_{nl}c_{nk}^\ast\frac{H_k(x+y)H_l(x-y)}{\sqrt{2^{k+l}k!l!\pi}} \\
& =\frac{1}{\pi}\int\limits_{-\infty}^\infty dy\,\mathrm{e}^{2ipy}\;\mathrm{e}^{-x^2-y^2}\, \sum_{n=1}^N\,p_n \\
& \times\sum_{k,l=0}^{d-1}c_{nl}c_{nk}^\ast \frac{H_k(x+y)H_l(x-y)}{\sqrt{2^{k+l}k!l!\pi}}.
\end{aligned} 
\end{equation}
Now $\frac{1}{\pi}\int\limits_{-\infty}^\infty dy\,\mathrm{e}^{2ipy}$ is just a Fourier transform operator with respect to $y$, and $\mathrm{e}^{-x^2-y^2}$ is a Gaussian function. Hence, for $W(x,p)$ being Gaussian, the sum $P:=\sum_{n=1}^N\,p_n \sum_{k,l=0}^{d-1}c_{nl}c_{nk}^\ast \frac{H_k(x+y)H_l(x-y)}{\sqrt{2^{k+l}k!l!\pi}}$ also has to be Gaussian in both $x$ and $y$. However, $P$ is just a polynomial. Strictly speaking, it is a polynomial of \textit{finite} order. Hence, $P$ can never be an exponential function and therefore also not become Gaussian. There is one exception though. $P$ may be a polynomial of order zero, i.e., a constant. 

So, we can assume that $\hat{\rho}$ is Gaussian and show that this is the case if and only if $c_{nl}=0\;\forall\;l\geq1$. "$\Leftarrow$" is trivial, since for $c_{nl}=0\;\forall\;n\,,\;\forall\;l\geq1$, $\hat{\rho}=\ket{0}\bra{0}$, which is obviously Gaussian. For "$\Rightarrow$" we perform an induction with respect to the dimension $d$: 

\textit{Inductive Basis}. For $d=2$ there is only one term $x^{2(d-1)}=x^{2}$ in the sum $P$, which comes from $H_{d-1}(x+y)H_{d-1}(x-y)=H_{1}(x+y)H_{1}(x-y)$. Its coefficient is $\tilde{c}_{1}\sum_{n=1}^N p_n|c_{n,1}|^2$, where $\tilde{c}_{1}\neq0$ is a real constant due to the Hermite polynomials. If the state is Gaussian, $\sum_{n=1}^N p_n|c_{n,1}|^2=0$. Since $p_n>0\;\forall\;n$ and $|c_{n,1}|\geq0\;\forall\;n$, from $\sum_{n=1}^N p_n|c_{n,1}|^2=0$ follows $c_{n,1}=0\;\forall\;n$. Hence, $c_{nl}=0\;\forall\;l\geq1$ ($l$ can be only $0$ or $1$ for $d=2$) is proved as a necessary condition for Gaussianity for $d=2$.

\textit{Inductive Step}. Assume validity of the statement for $d=m$ and consider $d=m+1$. Then, there is only one term $x^{2((m+1)-1)}=x^{2m}$ in the sum $P$, which comes from $H_{(m+1)-1}(x+y)H_{(m+1)-1}(x-y)=H_{m}(x+y)H_{m}(x-y)$. Again with $\tilde{c}_{m}\neq0$ being a real constant the coefficient is $\tilde{c}_{m}\sum_{n=1}^N p_n|c_{n,m}|^2$. If the state is Gaussian, $\sum_{n=1}^N p_n|c_{n,m}|^2=0$. Since $p_n>0\;\forall\;n$ and $|c_{n,m}|\geq0\;\forall\;n$, from $\sum_{n=1}^N p_n|c_{n,m}|^2=0$ follows $c_{n,m}=0\;\forall\;n$. Then, from the inductive hypothesis it is known that $c_{nl}=0\;\forall\;1\leq l\leq d-2=m-1$. Hence, $c_{nl}=0\;\forall\;1\leq l\leq d-1=m$, which proves the validity of the statement for $d=m+1$ and hence for all finite $m$. 

Concluding, a $d$-dimensional quantum state with finite $d$ and $d\geq2$, written in the form \eqref{eq:Lemmaqudit}, is Gaussian if and only if $c_{nl}=0\;\forall\;l\geq1$, which corresponds to $\hat{\rho}=\ket{0}\bra{0}$. Therefore, $\hat{\rho}$ is Gaussian if and only if $\hat{\rho}=\ket{0}\bra{0}$, which is only one-dimensional. Hence, any $d$-dimensional quantum state with finite $d$ and $d\geq2$ is non-Gaussian.
\end{proof}

\section{AUXILIARY CALCULATIONS}\label{Apx:AuxCal}
The derivation of Eq. \eqref{eq:noisyQbQm3} becomes clear, when writing the thermal photon noise channel with aid of an ancilla Hilbert space (the environment E), and joint unitary evolution on the overall Hilbert space. Finally, a trace operation over the ancilla Hilbert space has to be performed:
\begin{equation}
\hat{\rho}'{}^{AB}=\mathrm{tr}_E[\hat{U}^{BE}\;\bigl( \ket{\phi}^{AB}\bra{\phi}\otimes\hat{\rho}^E_{thermal}\bigr) \;\hat{U}^{{BE}^\dagger}],
\end{equation}
with

(1) beam splitter unitary $\hat{U}^{BE}=e^{\theta(\hat{a}_E^\dagger \hat{a}_B-\hat{a}_B^\dagger \hat{a}_E)}$,

(2) environmental thermal state \\ $\hat{\rho}^E_{thermal}=\sum_{n=0}^\infty \frac{\braket{n_{th}}^n}{(1+\braket{n_{th}})^{n+1}}\,\ket{n}_{E}\bra{n}$ \cite{Scully},

(3) mean thermal photon number $\braket{n_{th}}$, 

(4) and beam splitter transmissivity $\eta=\cos^2\theta$.
Note that for $\braket{n_{th}}=0$, the photon loss channel is obtained, which is therefore just a limiting case of this channel. 
For the calculation, the following relations have been exploited:
\begin{align} \label{equse1}
\hat{U}^{{BE}^\dagger}\hat{U}^{BE} & =\mathbb{1}, \\ \label{equse2}
\hat{U}^{BE}(\hat{a}^{E^\dagger})^n\hat{U}^{{BE}^\dagger} & =(\sqrt{\eta}\hat{a}^{E^\dagger}-\sqrt{1-\eta}\hat{a}^{B^\dagger})^n, \\ \label{equse3} \hat{U}^{BE}\hat{D}^B(\alpha)\hat{U}^{{BE}^\dagger} & =\hat{D}^B(\sqrt{\eta}\alpha)\hat{D}^E(\sqrt{1-\eta}\alpha), \\ \label{equse4}
\hat{U}^{BE}\ket{0}^B\ket{0}^E & =\ket{0}^B\ket{0}^E, \\
\ket{\alpha}^B &=\hat{D}^B(\alpha)\ket{0}^B, \\
\ket{n}^E &=\frac{1}{\sqrt{n!}}(\hat{a}^{E^\dagger})^n\ket{0}^E, 
\end{align}
Furthermore, the binomial identity has been utilized.

The calculation of Eq. \eqref{eq:noisyQbQm4} proceeds similarly. However, instead of exploiting $\hat{\rho}^E_{thermal}=\sum_{n=0}^\infty \rho_n^{th}\,\ket{n}_{E}\bra{n}$ in the Fock basis, we made use of the coherent state basis. The Glauber-Sudarshan $P$ representation of the thermal environment is given by \cite{Scully}
\begin{equation}
P_{\hat{\rho}_{thermal}}(\alpha,\alpha^\ast)=\frac{1}{\pi\braket{n_{th}}}\mathrm{e}^{-\frac{|\alpha|^2}{\braket{n_{th}}}}.
\end{equation}
Therefore,
\begin{equation}
\hat{\rho}_{thermal}=\frac{1}{\pi\braket{n_{th}}}\int_{\mathbb C} d^2\alpha\,\mathrm{e}^{-\frac{|\alpha|^2}{\braket{n_{th}}}}\ket{\alpha}\bra{\alpha}.
\end{equation}
Expressing the thermal state in this form, the action of the thermal channel on an element $\ket{\alpha}\bra{\beta}$ is
\begin{equation}
\begin{aligned}
\$_{thermal}(\ket{\alpha}^B\bra{\beta}) &=\frac{1}{\pi\braket{n_{th}}}\int_{\mathbb C} d^2\gamma\,\mathrm{e}^{-\frac{|\gamma|^2}{\braket{n_{th}}}}\mathrm{tr}_E[\hat{U}^{BE}\;\bigl( \ket{\alpha}^B\bra{\beta} \\ &\otimes\ket{\gamma}^E\bra{\gamma}\bigr) \;\hat{U}^{{BE}^\dagger}].
\end{aligned}
\end{equation}
With the relations \eqref{equse1}, \eqref{equse3}, \eqref{equse4} and $\hat{U}^{BE}\hat{D}^E(\gamma)\hat{U}^{{BE}^\dagger} =\hat{D}^E(\sqrt{\eta}\gamma)\hat{D}^B(-\sqrt{1-\eta}\gamma)$ we find 
\begin{equation}
\begin{aligned}
\$_{thermal}(\ket{\alpha}^B\bra{\beta}) & = \frac{1}{\pi\braket{n_{th}}}\int_{\mathbb C} d^2\gamma\,\mathrm{e}^{-\frac{|\gamma|^2}{\braket{n_{th}}}} \\
&\times \braket{\sqrt{1-\eta}\beta +\sqrt{\eta}\gamma|\sqrt{1-\eta}\alpha+\sqrt{\eta}\gamma} \\
&\times \ket{\sqrt{\eta}\alpha-\sqrt{1-\eta}\gamma}^B\bra{\sqrt{\eta}\beta-\sqrt{1-\eta}\gamma}.
\end{aligned}
\end{equation}
Inserting this expression in the overall state $\hat{\rho}'{}^{AB}=({\mathbb 1}^A\otimes\$_{thermal}^B)\ket{\psi}^{AB}\bra{\psi}$ with $\ket{\psi}^{AB}=\frac{1}{\sqrt{2}}\Bigl(\ket{0}^A\ket{\alpha}^B+\ket{1}^A\ket{-\alpha}^B\Bigr)$ yields Eq. \eqref{eq:noisyQbQm4}.
\end{appendix}

%\bibliography{labibio}

\end{document}